\documentclass{article}

\usepackage{graphicx} 
\graphicspath{ {./figures/} }
\usepackage{blindtext}  
\usepackage[margin=8em]{geometry}
\usepackage{setspace}
\usepackage[dvipsnames, table]{xcolor}
\usepackage{amssymb}
\usepackage{amsfonts}
\usepackage{amsthm}
\usepackage{amsmath}
\usepackage{dsfont}
\usepackage{bbold}
\usepackage{authblk}
\usepackage{tikz}
\usepackage{setspace}
\usepackage[
    style=alphabetic,
    citestyle=numeric-comp, style=alphabetic, sorting=anyt, backref=true, minnames=3, maxnames=12, giveninits=true, defernumbers=true, doi=false, url=false
    ]{biblatex}
\DefineBibliographyStrings{english}{
  backrefpage={Page },
  backrefpages={Pages }
}
\newbibmacro{string+doi}[1]{\iffieldundef{doi}{#1}{\href{https://dx.doi.org/\thefield{doi}}{#1}}}
\DeclareFieldFormat[article]{title}{\usebibmacro{string+doi}{\mkbibquote{#1}}}
\DeclareFieldFormat[incollection]{title}{\usebibmacro{string+doi}{\mkbibquote{#1}}}
\DeclareFieldFormat[inproceedings]{title}{\usebibmacro{string+doi}{\mkbibquote{#1}}}
\newbibmacro{string+url}[1]{\iffieldundef{url}{#1}{\href{\thefield{url}}{#1}}}
\DeclareFieldFormat[PhDThesis]{title}{\usebibmacro{string+url}{\mkbibquote{#1}}}

\usepackage[colorlinks=true, linkcolor=Periwinkle, citecolor=MidnightBlue, urlcolor=black, linktocpage=true]{hyperref}
\usepackage{fancyhdr}
\usepackage[shortlabels]{enumitem} 
\usepackage{braket}
\usepackage{algorithm}
\usepackage{algpseudocode}
\usepackage{setspace}
\usepackage{mathtools} 
\usepackage{multirow,multicol}
\usepackage{array} 
\newcolumntype{P}[1]{>{\centering\arraybackslash}p{#1}} 
\makeatletter
\newcommand{\thickhline}{%
    \noalign {\ifnum 0=`}\fi \hrule height 1.2  pt
    \futurelet \reserved@a \@xhline
} 
\makeatother 
\usepackage{etoolbox}
\usepackage{cleveref}
\usepackage{makecell}
\makeatletter
\newsavebox\myboxA
\newsavebox\myboxB
\newlength\mylenA
\usepackage{tcolorbox}

\newcommand*\oline[1]{%
  \vbox{%
    \hrule height 0.5pt
    \kern0.25ex
    \hbox{%
      \kern-0.1em
      \ifmmode#1\else\ensuremath{#1}\fi
      \kern-0.1em
    }
  }
}
\allowdisplaybreaks

\newtheorem{theorem}{Theorem}[section]
\newtheorem{lemma}[theorem]{Lemma}

\newtheorem{definition}[theorem]{Definition}

\crefname{theorem}{Theorem}{Theorems}
\crefname{lemma}{Lemma}{Lemmas}
\crefname{table}{Table}{Tables}
\crefname{figure}{Figure}{Figures}

\newcommand{\Tr}{\operatorname{Tr}}

\newcommand{\Rbb}{\mathbb{R}}

\newcommand{\OC}{\mathcal{O}}

\newcommand{\Var}{\operatorname{Var}}
\newcommand{\E}{\mathbb{E}}
\newcommand{\Pbb}{\mathbb{P}}
\newcommand{\Cov}{\operatorname{Cov}}

\newcommand{\bdot}{\scalebox{0.65}{\textbullet}}
\renewcommand{\vec}[1]{\boldsymbol{#1} }
\newcommand{\mat}[1]{\boldsymbol{#1} }

\DeclarePairedDelimiter{\norm}{\|}{\|}
\newcommand{\ols}[1]{\mskip.5\thinmuskip\overline{\mskip-.5\thinmuskip {#1} \mskip-.5\thinmuskip}\mskip.5\thinmuskip} 

\renewcommand*{\thefootnote}{(\arabic{footnote})}

\newcommand\blfootnote[1]{%
  \begingroup
  \renewcommand\thefootnote{}\footnote{#1}%
  \addtocounter{footnote}{-1}%
  \endgroup
} 

\begin{document}

\title{Randomized and quantum approximate matrix multiplication}

\author[1]{Simon Apers}
\author[2]{Arjan Cornelissen}
\author[3]{Samson Wang}

\affil[1]{Universit\'e Paris Cit\'e, CNRS, IRIF, Paris, France}
\affil[2]{Simons Institute, UC Berkeley, California, USA}
\affil[3]{Institute for Quantum Information and Matter, Caltech, California, USA}

\maketitle

\begin{abstract}
The complexity of matrix multiplication is a central topic in computer science.
While the focus has traditionally been on exact algorithms, a long line of literature also considers randomized algorithms, which return an approximate solution in faster time. 
In this work, we adopt a unifying perspective that frames these randomized algorithms in terms of mean estimation.
Using it, we first give refined analyses of classical algorithms based on random walks by Cohen-Lewis (`99), and based on sketching by Sarlós (`06) and Drineas-Kannan-Mahoney (`06).
We then propose an improvement on Cohen-Lewis that yields a single classical algorithm that is faster than all the other approaches, if we assume no use of (exact) fast matrix multiplication as a subroutine.
Second, we demonstrate a quantum speedup on top of these algorithms by using the recent quantum multivariate mean estimation algorithm by Cornelissen-Hamoudi-Jerbi (`22).


\end{abstract}

\blfootnote{*All authors contributed equally and are listed in alphabetical order.}





\vspace{0.5em}
\section{Introduction}

Over the last 6 decades, the complexity of multiplying two $n$-by-$n$ matrices has gone down from the na\"ive $O(n^3)$ upper bound to $O(n^\omega)$ with $\omega < 2.372$ the matrix multiplication coefficient -- see \cite{alman2025more} for the latest improvement.
Whether this coefficient can be brought down to its trivial optimum $\omega=2$ is a major open question in computer science.
In the meantime, however, the complexity can be reduced by considering algorithms for \emph{approximate} matrix multiplication. 

\begin{definition}[Approximate matrix multiplication]\label{def:problem}
Given a description of matrices $\mat{A}$ and $\mat{B}$, return a matrix $\mat{C}$ which satisfies $\| \mat{C} - \mat{AB} \|_* \leq \varepsilon$ for some chosen norm $\|\cdot \|_*$.
\end{definition}

\noindent
In our work we specifically consider $\|\cdot\|_* = \|\cdot\|_{\max}$ (the ``max-norm problem'') and $\|\cdot\|_* = \|\cdot\|_F$ (the ``Frobenius norm problem'').

Approximate matrix multiplication is a setting that allows to exploit \emph{randomized} algorithms, and this seems in contrast to exact matrix multiplication where the best known algorithms are deterministic.
In particular, there are well-known Monte Carlo algorithms for approximate matrix multiplication that build on random walks (or ``path integral Monte Carlo'') \cite{cohen1999approximating} and sketching \cite{sarlos2006improved,drineas2006fast}.
These algorithms can ultimately be viewed from the perspective of \emph{multivariate mean estimation}, i.e., they cleverly define a random variable $\mat{X}$, taking values in $\Rbb^d$ with $d = n^2$, whose mean $\vec{\mu}$ contains the entries of the matrix product $\vec{AB}$. Using e.g.~the vector Bernstein inequality, one can argue that with high probability, the empirical mean $\tilde{\vec{\mu}}$ of $L$ samples of $\mat{X}$, with covariance matrix $\mat{\Sigma} := \E[(\vec{X} - \E[\vec{X}])(\vec{X} - \E[\vec{X}])^T]$, satisfies
\begin{equation}
    \label{eq:classical-mean-estimation}
    \textit{multivariate mean estimation:} \quad
    \left\{ \begin{array}{l}
        \| \tilde{\vec{\mu}} - \vec{\mu} \|_{\infty} = \OC \left( \sqrt{\max_i [\mat{\Sigma}]_{ii}/L} \right) \\
        \| \tilde{\vec{\mu}} - \vec{\mu} \|_2 = \OC \left( \sqrt{\Tr[\mat{\Sigma}]/L} \right)
    \end{array}\right..
\end{equation}

The fact that these randomized algorithms can be phrased as a mean estimation task means that they are amenable to a quantum speedup, since we can make use of \emph{quantum multivariate mean estimation} from \cite{cornelissen2022multiVariateMeanEstimation}: given quantum access to $\widetilde{O}(L)$ samples of $\mat{X}$, we can return an estimator $\tilde{\vec{\mu}}$ that with high probability satisfies
\begin{equation}
    \label{eq:quantum-mean-estimation}
    \textit{quantum multivariate mean estimation:} \quad
    \left\{\begin{array}{l}
        \| \tilde{\vec{\mu}} - \vec{\mu} \|_{\infty} = \OC \left( \sqrt{\Tr[\mat{\Sigma}]}/L \right) \\
        \| \tilde{\vec{\mu}} - \vec{\mu} \|_2 = \OC \left( \sqrt{d \Tr[\mat{\Sigma}]}/L \right)
    \end{array}\right..
\end{equation}
For the $\ell_2$-norm, for instance, this yields a speedup over classical mean estimation when the number of quantum samples $L \gg d = n^2$. 

In this work, we study classical and quantum algorithms for approximate matrix multiplication utilizing mean estimation.
We make the following contributions:

\begin{itemize}
    \item We refine the seminal algorithm of Cohen \textit{\&} Lewis \cite{cohen1999approximating}. While the original algorithm is incomparable to the main other approach based on matrix sketching \cite{sarlos2006improved, drineas2006fast}, our refinement yields a classical algorithm with better time complexity than prior classical art (if no fast matrix multiplication is allowed).
    \item We give a unified analysis of sketching algorithms for matrix multiplication \cite{sarlos2006improved, drineas2006fast} based on mean estimation. Using this, we construct quantum analogues of all above algorithms. We find a quantum speedup in certain parameter regimes, and there is an unconditional speedup over prior quantum art.
    \item We extend our sketching-based analysis to give classical and quantum matrix multiplication algorithms for more than two matrices. 
\end{itemize}

We display a simplified account of our results in Table \ref{tab:results-introduction} (a more intricate comparison of time complexities can be found in Section \ref{sec:complexity-comparison}).
We express the runtimes using the element-wise matrix norm: for $\mat{M} \in \Rbb^{n \times m}$ and $p,q \in [1,\infty]$, we write
    \[\norm{\mat{M}}_{p,q} = \left(\sum_{j=1}^n \left(\sum_{k=1}^m \big|[\mat{M}]_{jk}\big|^p\right)^{\frac{q}{p}}\right)^{\frac1q}.\]
Specifically, $\norm{\cdot}_{\max} = \norm{\cdot}_{\infty,\infty}$ and $\norm{\cdot}_F = \norm{\cdot}_{2,2}$ are the max-norm and Frobenius norm, respectively (see Section \ref{sec:lin-algebra} for more details).
In order to facilitate navigation of the table, we also present the partial ordering of all norms that appear in Figure \ref{fig:norms}.

\begin{table}[!ht]
\renewcommand{\arraystretch}{2}
    \centering
    \begin{tabular}{rr|cc|c|c}\label{tab:results-introduction}
        &&  \multicolumn{2}{c|}{Runtime\ $(\widetilde{O}(\cdot))$} & \multirow{2}{*}{\makecell{Data model}} & \multirow{2}{*}{\makecell{Multiple \\ matrices?}} \\
        && $\|\cdot\|_{\max}$ & $\|\cdot\|_F$&    \\\hline
        \multirow{4}{*}{\rotatebox{90}{Classical}} & Cohen-Lewis \cite{cohen1999approximating} &  
        $\frac{\norm{\ols{\mat{A}} \ols{\mat{B}}}_{1,2}^2}{\varepsilon^2}$ 
        & $n\frac{\norm{\ols{\mat{A}} \ols{\mat{B}}}_{1,2}^2}{\varepsilon^2}$ & RAM &\checkmark \\
        & Sarl\'os \cite{sarlos2006improved} & $ n^2 \frac{\norm{\mat{A}}_{2,\infty}^2\norm{\mat{B}^T}_{2,\infty}^2}{\varepsilon^2}$ & $n^2 \frac{\norm{\mat{A}}_{2,2}^2\norm{\mat{B}}_{2,2}^2}{\varepsilon^2}$ & circuit & \checkmark \\
        &  \makecell[r]{Drineas-Kannan-Mahoney \\ {\cite{drineas2006fast}}}  & $ n^2 \frac{\norm{\mat{A}^T}_{\infty,2}^2\norm{\mat{B}}_{\infty, 2}^2}{\varepsilon^2} $ & $n^2 \frac{\norm{\mat{A}}_{2,2}^2\norm{\mat{B}}_{2,2}^2}{\varepsilon^2}$ & circuit &  \\
        & \makecell[r]{ \bf{This work}  \\ (random walk)} & $\frac{\norm{\mat{\ols{A}}\mat{\ols{B}}}_{\infty,\infty} \norm{\mat{\ols{A}}\mat{\ols{B}}}_{1,1}}{\varepsilon^2}$ & \cellcolor{Periwinkle!15}$\frac{\norm{\ols{\mat{A}} \ols{\mat{B}}}_{1,1}^2}{\varepsilon^2}$ & RAM & \checkmark
        \\\hline
        \multirow{3}{*}{\rotatebox{90}{Quantum}}  & Shao \cite{shao2018quantum} & $\frac{\norm{\mat{A}}_{2,1}\norm{\mat{B}^T}_{2,1}}{\varepsilon}$ & $n \frac{\norm{\mat{A}}_{2,1}\norm{\mat{B}^T}_{2,1}}{\varepsilon}$ & QROM & \\
        & \makecell[r]{\bf This work  \\ (quantum tug-of-war sketch)} & $n^2 \frac{\norm{\mat{A}}_{2,2}\norm{\mat{B}}_{2,2}}{\varepsilon}$ & & \makecell{quantum \\ circuit} & \checkmark \\
        & \makecell[r]{\bf This work \\ (quantum random walk)} & \cellcolor{Periwinkle!20}$\frac{\norm{\ols{\mat{A}} \ols{\mat{B}}}_{1,1}}{\varepsilon}$  & \cellcolor{Periwinkle!20}$n \frac{\norm{\ols{\mat{A}} \ols{\mat{B}}}_{1,1}}{\varepsilon}$ & QRAM &\checkmark 
    \end{tabular}
    \caption{\textbf{Time complexities of approximate matrix multiplication (without use of fast matrix multiplication).} For simplicity of presentation, we assume square matrices $\mat{A}$, $\mat{B} \in \mathbb{R}^{n \times n}$, and a matrix multiplication exponent of $\omega=3$. We use an overline over any matrix to denote the matrix of element-wise absolute values. The element-wise norms we use are defined in Section \ref{sec:lin-algebra}, and we present their partial ordering in Figure \ref{fig:norms}. If within a block (i.e., ``Classical" or ``Quantum") and a particular column there is an algorithm which is unconditionally faster than the others, we have shaded its complexity. Each of the algorithms utilizes $O(n^2)$ preprocessing time which we omit from the table for simplicity. We discuss the different (classical and quantum) data models in Section \ref{sec:computational-model}. We present a more detailed table of complexities in Table \ref{tab:results-full}. The cell left open would have a complexity that exceeds $O(n^3)$, and hence is not very useful.}
    \label{tbl:complexities}
\end{table}

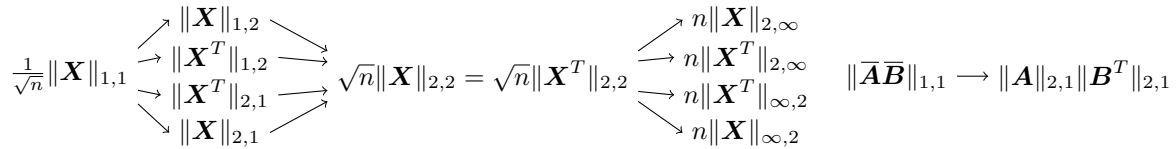
\begin{figure}[!ht]
    \centering
    \begin{tikzpicture}[xscale=2, yscale=.5]
        \node (11) at (0,0) {$\frac{1}{\sqrt{n}}\|\mat{X}\|_{1,1}$};
        
        \node (12) at (1,1.5) {$\|\mat{X}\|_{1,2}$};
        \node (12T) at (1,.5) {$\|\mat{X}^T\|_{1,2}$};
        \node (21T) at (1,-.5) {$\|\mat{X}^T\|_{2,1}$};
        \node (21) at (1,-1.5) {$\|\mat{X}\|_{2,1}$};
    
        \draw[->] (11.north east) to (12.west);
        \draw[->] (11) to (12T.west);
        \draw[->] (11) to (21T.west);
        \draw[->] (11.south east) to (21.west);

        \node (22) at (2.75,0) {$\sqrt{n}\|\mat{X}\|_{2,2} = \sqrt{n}\|\mat{X}^T\|_{2,2}$};

        \draw[->] (12.east) to (22.north west);
        \draw[->] (12T.east) to (22);
        \draw[->] (21T.east) to (22);
        \draw[->] (21.east) to (22.south west);

        \node (2inf) at (4.5,1.5) {$n\|\mat{X}\|_{2,\infty}$};
        \node (2infT) at (4.5,.5) {$n\|\mat{X}^T\|_{2,\infty}$};
        \node (inf2T) at (4.5,-.5) {$n\|\mat{X}^T\|_{\infty,2}$};
        \node (inf2) at (4.5,-1.5) {$n\|\mat{X}\|_{\infty,2}$};

        \draw[->] (22.north east) to (2inf.west);
        \draw[->] (22) to (2infT.west);
        \draw[->] (22) to (inf2T.west);
        \draw[->] (22.south east) to (inf2.west);

        \node (AB) at (5.5,0) {$\|\mat{\ols{A}}\mat{\ols{B}}\|_{1,1}$};
        \node (AtimesB) at (6.75,0) {$\|\mat{A}\|_{2,1}\|\mat{B}^T\|_{2,1}$};

        \draw[->] (AB) to (AtimesB);
    \end{tikzpicture}
    \caption{Partial ordering of matrix norms appearing in Table \ref{tab:results-introduction}. Here, $\mat{X}, \mat{A}, \mat{B} \in \Rbb^{n \times n}$, and $\mat{\ols{A}}$ represents the matrix containing all element-wise absolute values of $\mat{A}$. ``$\rightarrow$" denotes ``$\leq$". If there is no arrow in between two quantities then they are incomparable.}
    \label{fig:norms}
\end{figure}

To our knowledge the best previous quantum algorithm is that of Shao \cite{shao2018quantum}, which operates via a quantum subroutine to perform inner product estimation on each of the $n^2$ inner products that determine a matrix product.\footnote{See also the algorithm \cite{bernasconi2024quantum} which can also be adapted to give an algorithm of similar complexity, though it is specialized to outputting a quantum state.} The data model Shao requires is efficient access to the unitaries which encode the rows and columns of $\mat{A}$ and $\mat{B}$ in the amplitudes of quantum states -- this is achievable in $O(n^2)$ preprocessing time and $\textrm{polylog}(n)$ time per query with a QROM \cite{prakash2014QLinAlgAndMLThesis} (we define and discuss quantum memory models in Section \ref{sec:computational-model}). We note that Shao also presents other algorithms which depend on the condition number of the matrix which we do not display in our table. Whilst these approaches could be advantageous in certain settings, we focus on comparing algorithms for generic matrices which are efficient without additional assumptions.

We remark that our algorithms with the most refined time complexities inherit a particular property of the algorithm of Cohen \textit{\&} Lewis \cite{cohen1999approximating} --- that the complexity depends on the norm of a matrix product, rather than a product of matrix norms. This can be especially advantageous for non-negative matrices or matrices with (approximate) sparsity patterns, and this advantage can be amplified in generalizations to multiple matrix products. As an example, when $\mat{A}$ and $\mat{B}$ are stochastic matrices, all Frobenius norm algorithms are slower than their counterparts based on \cite{cohen1999approximating} by a factor up to $\Omega(n)$ . 


\subsection{Random walk algorithms}

While the Cohen-Lewis algorithm \cite{cohen1999approximating} works for arbitrary matrices, its idea is best illustrated for the multiplication of stochastic matrices $\mat{A}_1, \dots, \mat{A}_k$, i.e., every matrix $\mat{A}_j \in \Rbb^{n \times n}$ is entry-wise non-negative, and all row-sums equal $1$. The matrix product $\mat{A}_1, \dots, \mat{A}_k$ can then be interpreted as the probability transition matrix of a random walk on a directed graph, displayed in \cref{fig:random-walk}. Simulting this random walk samples from the distribution $\vec{d}^T\mat{A}_1 \cdots \mat{A}_k$, for an arbitrary initial distribution $\vec{d} \in \Rbb^n$. The original idea from Cohen \& Lewis is to reconstruct the matrix product $\mat{A}_1 \cdots \mat{A}_k$ by learning it row-by-row, i.e., learning all the probability distributions by starting from $\vec{d} = \vec{e}_j$, with $j$ running from $1$ to $n$.

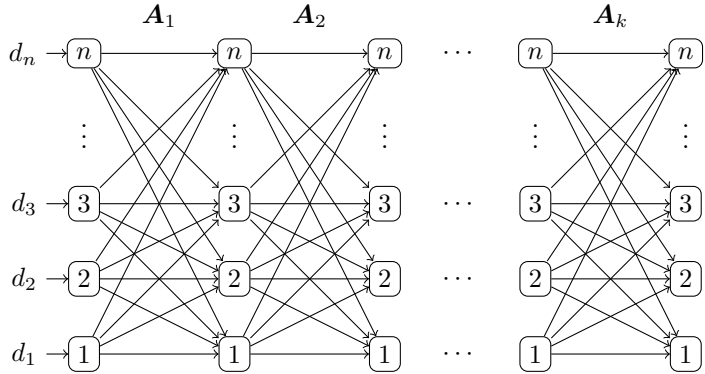
\begin{figure}[!ht]
    \centering
    \begin{tikzpicture}[vertex/.style = {draw, rounded corners = .3em}, xscale=2]
        \node[vertex] (01) at (0,1) {$1$};
        \node[vertex] (02) at (0,2) {$2$};
        \node[vertex] (03) at (0,3) {$3$};
        \node at (0,4) {$\vdots$};
        \node[vertex] (0n) at (0,5) {$n$};
        \draw[<-] (01) to (-.25,1) node[left] {$d_1$};
        \draw[<-] (02) to (-.25,2) node[left] {$d_2$};
        \draw[<-] (03) to (-.25,3) node[left] {$d_3$};
        \draw[<-] (0n) to (-.25,5) node[left] {$d_n$};
        
        \node[vertex] (11) at (1,1) {$1$};
        \node[vertex] (12) at (1,2) {$2$};
        \node[vertex] (13) at (1,3) {$3$};
        \node at (1,4) {$\vdots$};
        \node[vertex] (1n) at (1,5) {$n$};
        \draw[->] (01) to (11);
        \draw[->] (01) to (12);
        \draw[->] (01) to (13);
        \draw[->] (01) to (1n);
        \draw[->] (02) to (11);
        \draw[->] (02) to (12);
        \draw[->] (02) to (13);
        \draw[->] (02) to (1n);
        \draw[->] (03) to (11);
        \draw[->] (03) to (12);
        \draw[->] (03) to (13);
        \draw[->] (03) to (1n);
        \draw[->] (0n) to (11);
        \draw[->] (0n) to (12);
        \draw[->] (0n) to (13);
        \draw[->] (0n) to (1n);
        \node at (.5, 5.5) {$\mat{A}_1$};
        \node[vertex] (21) at (2,1) {$1$};
        \node[vertex] (22) at (2,2) {$2$};
        \node[vertex] (23) at (2,3) {$3$};
        \node at (2,4) {$\vdots$};
        \node[vertex] (2n) at (2,5) {$n$};
        \draw[->] (11) to (21);
        \draw[->] (11) to (22);
        \draw[->] (11) to (23);
        \draw[->] (11) to (2n);
        \draw[->] (12) to (21);
        \draw[->] (12) to (22);
        \draw[->] (12) to (23);
        \draw[->] (12) to (2n);
        \draw[->] (13) to (21);
        \draw[->] (13) to (22);
        \draw[->] (13) to (23);
        \draw[->] (13) to (2n);
        \draw[->] (1n) to (21);
        \draw[->] (1n) to (22);
        \draw[->] (1n) to (23);
        \draw[->] (1n) to (2n);
        \node at (1.5, 5.5) {$\mat{A}_2$};
        
        \node at (2.5,1) {$\cdots$};
        \node at (2.5,2) {$\cdots$};
        \node at (2.5,3) {$\cdots$};
        \node at (2.5,5) {$\cdots$};
        
        \node[vertex] (31) at (3,1) {$1$};
        \node[vertex] (32) at (3,2) {$2$};
        \node[vertex] (33) at (3,3) {$3$};
        \node at (3,4) {$\vdots$};
        \node[vertex] (3n) at (3,5) {$n$};
        \node[vertex] (41) at (4,1) {$1$};
        \node[vertex] (42) at (4,2) {$2$};
        \node[vertex] (43) at (4,3) {$3$};
        \node at (4,4) {$\vdots$};
        \node[vertex] (4n) at (4,5) {$n$};
        \draw[->] (31) to (41);
        \draw[->] (31) to (42);
        \draw[->] (31) to (43);
        \draw[->] (31) to (4n);
        \draw[->] (32) to (41);
        \draw[->] (32) to (42);
        \draw[->] (32) to (43);
        \draw[->] (32) to (4n);
        \draw[->] (33) to (41);
        \draw[->] (33) to (42);
        \draw[->] (33) to (43);
        \draw[->] (33) to (4n);
        \draw[->] (3n) to (41);
        \draw[->] (3n) to (42);
        \draw[->] (3n) to (43);
        \draw[->] (3n) to (4n);
        
        \node at (3.5, 5.5) {$\mat{A}_k$};
    \end{tikzpicture}
    \caption{A pictorial representation of the matrix product $\mat{A}_1 \cdots \mat{A}_k$ of stochastic matrices as a random walk on a directed graph. The random walk samples a vertex on the left side of the graph with probability distribution $\vec{d}$, and then transitions through the graph, with the probability transition matrix of every layer equal to $\mat{A}_j$. After $k$ steps, the probability distribution on the final layer of vertices is $\vec{d}^T\mat{A}_1 \cdots \mat{A}_k$.}
    \label{fig:random-walk}
\end{figure}

We subtly modify this approach, by analyzing what happens when we start with an arbitrary initial distribution $\vec{d}$. We denote the resulting sampled path of vertices $(j_0, \dots, j_k)$, and define the corresponding random variable $\mat{X}(j_0, \dots, j_k) = \vec{e}_{j_0}\vec{e}_{j_k}^T$. We then prove in \cref{lem:CL-expectation} that this can be used as an estimator for the matrix product, as
\[\E[\mat{X}] = \mathrm{diag}(\vec{d}) \cdot \mat{A}_1 \cdots \mat{A}_k\,, \qquad \text{and} \qquad \Tr[\mat{\Sigma}] \leq 1\,,\]
where we denote $\mathrm{diag}(\vec{d})$ as the diagonal matrix who shares diagonal entries with $\vec{d}$.

Next, we adapt the above approach to work for matrices with negative entries as well, i.e., we explain how to compute $\mat{A}_1 \cdots \mat{A}_k$, where $\mat{\ols{A}}_1, \dots, \mat{\ols{A}}_k$ are all stochastic matrices. The idea here, already observed by Cohen \& Lewis, is to label all the edges of the directed graph with the signs of the matrix entries. For a sampled path $(j_0, \dots, j_k)$, we then multiply all the signs together to obtain $s(j_0, \dots, j_k) \in \{\pm 1\}$, and we modify the random variable to be $\mat{X}(j_0, \dots, j_k) = s(j_0, \dots, j_k)\vec{e}_{j_0}\vec{e}_{j_k}^T$. Crucially, this does not impact the upper bound on the trace of the covariance matrix, since we can still show that $\Tr[\mat{\Sigma}] \leq 1$.

Finally, we adapt the construction to any matrix product, by observing that for any sequence of matrices $\mat{A}_1, \dots, \mat{A}_k$, we can construct a \textit{stochastic matrix product decomposition}. That is, we can find a vector $\vec{d}$ and a sequence of matrices $\mat{A}_1', \dots, \mat{A}_k'$, such that $\mat{\ols{A}}_1', \dots, \mat{\ols{A}}_k'$ are all stochastic, and such that $\mat{A}_1 \cdots \mat{A}_k = \mathrm{diag}(\vec{d})\mat{A}_1' \cdots \mat{A}_k'$ (see \cref{thm:matrix-product-decomposition}). Moreover, we prove that we can perform this stochastic matrix product decomposition as a preprocessing routine in time linear in the size of the input. We can now interpret the vector $\mat{d}$ as the probability distribution $\vec{d}/\|\vec{d}\|_1$, and absorb the resulting factor $\norm{\vec{d}}_1$ in the random variable, i.e., $\mat{X}(j_0, \dots, j_k) = \norm{\vec{d}}_1 \cdot s(j_0, \dots, j_k) \cdot \vec{e}_{j_0}\vec{e}_{j_k}^T$. We then find that
\[\E[\mat{X}] = \mat{A}_1 \cdots \mat{A}_k\,, \quad \max_{ij} [\mat{\Sigma}]_{ii} \leq \|\mat{A}_1 \cdots \mat{A}_k\|_{\infty,\infty} \cdot \|\mat{A}_1 \cdots \mat{A}_k\|_{1,1}\,, \quad \text{and} \quad \Tr[\mat{\Sigma}] \leq \|\mat{A}_1 \cdots \mat{A}_k\|_{1,1}^2\,.\]
Plugging these expressions into Eqs.~\eqref{eq:classical-mean-estimation}, \eqref{eq:quantum-mean-estimation} results in the complexities in \cref{tbl:complexities}.

%

To actually implement these algorithms, we require efficiently addressable read/write-memory in both the classical and quantum settings. Putting all these ideas together gives the following result.

\begin{theorem}[Simplified version of Theorems \ref{thm:improved-cohen-lewis} \textit{\&} \ref{thm:quantum-cohen-lewis}]\label{thm:informal-cohen-lewis}
    Consider a matrix product of $k$ matrices $\mat{A}_1 \cdots \mat{A}_k$ and the problem in Definition \ref{def:problem}. We use an overline over any matrix to denote the matrix of element-wise absolute values. We give classical algorithms in the RAM model, with run-times
    \[\widetilde{O}\left(k \cdot \frac{\norm{\mat{\ols{A}}_1 \cdots \mat{\ols{A}}_k}_{\infty,\infty} \cdot \norm{\mat{\ols{A}}_1 \cdots \mat{\ols{A}}_k}_{1,1}}{\varepsilon^2} \right), \qquad \text{and} \qquad \widetilde{O}\left(k \cdot \frac{\norm{\ols{\vec{A}_1} \cdots \ols{\vec{A}_k}}_{1,1}^2}{\varepsilon^2} \right),\]
    that solve the $\max$-norm and the Frobenius-norm problems, respectively, with high probability.  Similarly, we give quantum algorithms in the QRAM-model, with run-times
    \[\widetilde{O}\left(k \cdot \frac{\norm{\ols{\vec{A}_1} \cdots \ols{\vec{A}_k}}_{1,1}}{\varepsilon}\right), \qquad \text{and} \qquad \widetilde{O}\left(k \cdot \sqrt{nm} \frac{\norm{\ols{\vec{A}_1} \cdots \ols{\vec{A}_k}}_{1,1}}{\varepsilon}\right),\]
    that solve the $\max$-norm and the Frobenius-norm problems, respectively, with high probability. Here, $\mat{A}_1$ has $n$ rows, and $\mat{A}_k$ has $m$ columns. Both algorithms require classical preprocessing time linear in the size of the input up to polylogarithmic factors.
\end{theorem}

The quantum complexity for Frobenius norm approximation which we demonstrate above simply follows by considering standard norm conversion between the max-norm and the Frobenius norm, and rescaling the error parameter. 

We remark that aside from being unconditionally faster than other algorithms in certain norm regimes (and the absence of fast matrix multiplication), the algorithms based on random walks are particularly well-suited to stochastic matrices. In particular, their complexities grow linearly in the number of matrices, rather than geometrically as is the case for all other approaches.

\subsection{Sketching algorithms}

The key idea of sketching algorithms as used for matrix multiplication is to sample vectors $\vec{s} \in \mathbb{R}^n$ and to compute pairs of vectors $\mat{A}\vec{s}$ and $\vec{s}^T\mat{B}$, which together form ``sketches" of $\mat{A}$ and $\mat{B}$, respectively. With judicious choice of $\vec{s}$, we can ensure first that the outer product $\mat{A}\vec{s}\vec{s}^T\mat{B}$ forms an unbiased estimator of $\mat{AB}$, and second that the second-order moments of this estimator are controlled. Sarl\'os \cite{sarlos2006improved} and Drineas et al.~\cite{drineas2006fast} propose different constructions for $\vec{s}$, and we build on these for our generalization to multiple matrices, as well as for our quantum algorithms. We will refer to their approaches as the tug-of-war sketch and column-sample sketch, respectively.

The tug-of-war sketch \cite{alon1996spaceComplexityFrequencyMoments} (also known as the AMS sketch) populates entries of $\vec{s}$ with a $4$-wise independent hash function $h:[n] \rightarrow\{-1,1\}$  the resulting sketched rows of $\mat{A}$ and sketched columns of $\mat{B}$ are ``pulled" in either the positive or negative directions with equal probability, hence the name ``tug-of-war sketch". The column-sample sketch samples a column of $\mat{A}$ (and correspondingly, a row of $\mat{B}$) according to a some probability distribution of choice. Thus, $\vec{s}$ takes the form of a 1-sparse vector weighted by the chosen probability distribution. Utilizing these sketches we give our results on sketching-based matrix multiplication algorithms.

\begin{theorem}[Simplified version of Theorems \ref{thm:Sarlos-multi} and \ref{thm:Sarlos-multi-fast}]
    Consider a matrix product of $k$ square matrices  $\mat{A}_1 \cdots \mat{A}_k \in \mathbb{R}^{n \times n}$ and the problem in Definition \ref{def:problem}. With classical preprocessing time linear in the size of the input up to polylogarithmic factors, we give:
    \[\makecell[l]{\text{a classical algorithm for the Frobenius}\\ \text{norm problem with run-time}} \quad\; \widetilde{O}\left( n^2 \cdot  \left(3^k\frac{\norm{\vec{A}_1}_{2,2}^2 \cdots \norm{\vec{A}_k}_{2,2}^2}{\varepsilon^2}\right)^{\omega -2} \right), \]
    \[\text{ a quantum algorithm for the max-norm problem with run-time \;} \widetilde{O}\left(n^2 \cdot 3^k\frac{\norm{\vec{A}_1}_{2,2} \cdots \norm{\vec{A}_k}_{2,2}}{\varepsilon} \right) \,, \]
    with high probability, where $\omega$ is the matrix multiplication exponent. Both algorithms do not require efficiently-addressable memory, and can run in the circuit model.
\end{theorem}

The complexity of our classical algorithm we state above reduces to that of \cite{sarlos2006improved} and \cite{drineas2006fast} in the case of multiplication of two matrices. To the authors' knowledge this property is not shared by previous analyses of sketching algorithms for multiple matrices  -- they either only naturally extend to 3 matrices \cite[Appendix A1]{drineas2006fast} or incur large (super-quadratic) runtime in the inverse precision $\varepsilon^{-1}$ \cite[Appendix B2]{sarlos2006improved}.  

The classical algorithm can be sped up by use of fast matrix multiplication, as can all the sketching-based classical algorithms we consider in this manuscript. 
Interestingly, for any matrix multiplication exponent $\omega < 2.5$, this yields an algorithm with time complexity scaling sublinearly with $\varepsilon^{-1}$, which has better dependence on the precision even over quantum algorithms. The origin of the advantage using fast matrix multiplication can be viewed as follows. Conventionally, we might wish to estimate $\E[\mat{A}\vec{s} \vec{s}^T\mat{B}] = \mat{A}\mat{B}$ by taking a sample mean of $c$ independently drawn samples $\{\vec{s}_j \}_j$. Naturally, the time complexity is linear in $c$. However, due to the structure of the outer product, such a sample mean is also exactly equal to $\mat{A}\vec{S}^T \vec{S}\mat{B}$ where the rows of the matrix $\mat{S} \in \mathbb{R}^{c \times n}$ are given by $\{\vec{s}_j /\sqrt{c}\}_j$. The matrix $\mat{A}\mat{S}^T \mat{S}\mat{B}$ can now be evaluated with fast matrix multiplication with time complexity proportional to $c^{\omega-2}$. This provides a way to return the sample mean sublinearly (after preprocessing time proportional to $nc$ to construct $\mat{S}$). 

Our quantum algorithm as quoted above is based on the tug-of-war sketch utilized by Sarl\'os. We can also give a quantum algorithm with similar complexity based on the column-sample sketch of Drineas et al., however, unlike the algorithms based on the tug-of-war sketch both classical and quantum approaches can only accommodate multiplication of two matrices. We note that the quantum algorithm we state has unconditionally worse time complexity than the one we give in Theorem \ref{thm:informal-cohen-lewis}. However, we record it as it uses a weaker assumption on the type of data structure required, which we discuss in Section \ref{sec:computational-model} (and it can handle multiple matrices which the algorithm of \cite{shao2018quantum} cannot).






\subsection{Discussion}




Our quantum algorithms for max-norm approximation depend on the trace of the covariance matrix, which is the same quantity that characterizes approximation in Frobenius norm. We note that in classical mean estimation, there is a different quantity that controls max-norm approximations: the maximum entry-wise variance --- as clear examples, note the disparity in the complexity comparison of \cref{tbl:complexities}. It is an interesting open question if the quantum multivariate mean estimation of \cite{cornelissen2022multiVariateMeanEstimation} in max-norm can be refined to depend on the same second-moment quantity that we see for classical mean estimation.

In the presentation of Table \ref{tab:results-introduction}
we have made the simplifying assumption that the matrix multiplication exponent is $\omega = 3$. In reality, one can choose $\omega < 3$ by exploiting fast matrix multiplication, in which case there are interesting asymptotic conclusions about the runtime of sketching algorithms --- they can outperform the quantum approaches we present in certain regimes. We remark that the most costly step of our quantum algorithms essentially is matrix-vector multiplication in coherent arithmetic (rather than rectangular matrix multiplication in classical sketching algorithms), and we leave it as an open question whether or not fast matrix multiplication can be exploited in our quantum algorithms. Meanwhile, in comparing our classical algorithm to those which can be sped up with fast matrix multiplication, we do stress that we may be more interested in practical considerations for immediate application. From this perspective, we recall that the Strassen algorithm \cite{strassen1969gaussian} with $\omega \approx 2.8$ is commonly implemented. However, refinements based around the work of Coppersmith and Winograd beyond $\omega < 2.38$ \cite{coppersmith1990matrixMultiplication} are considered ``galactic", and not practically implementable \cite{iliopoulos1989worstCaseMultiplication}.

Finally, our work considers generic unstructured matrices. However, other classical approaches to approximate matrix multiplication are particularly useful in exploiting sparse matrices \cite{pagh2013compressed}. We leave open the question of whether we can exploit sparsity in the quantum setting as well.

\section{Preliminaries}

\subsection{Notation}

We denote matrices and vectors in bold type $\mat{A} \in \mathbb{C}^{N\times N}$ and $\mat{a} \in \mathbb{C}^{N}$. We denote their entries as $[\mat{A}]_{ij}$ and $a_{j}$ respectively for all $i, j \in [N]$. Denote $[\mat{A}]_{k \bdot} \in \mathbb{C}^{1 \times N}$ as the $k^{th}$ row matrix of $\mat{A}$ for all $k \in [N]$, and denote $[\mat{A}]_{\bdot j} \in \mathbb{C}^{N\times 1}$ as the $k^{th}$ column matrix of $\mat{A}$ for all $j \in [N]$. For any matrix, we write $\ols{\mat{A}}$ as the matrix with element-wise absolute values $[\ols{\mat{A}}]_{ij} = |[\mat{A}]_{ij}|$. We denote $\mat{I}$ as the identity matrix and we denote $\vec{1}_m$ as the all-ones vector of dimension $m$. 

The notation $\widetilde{O}(\cdot)$ hides polylogarithmic factors, that is we denote $\widetilde{O}(g(n)) = O\big(g(n)\log^k(g(n))\big)$ for some constant $k$.

\subsection{Linear algebra and probability theory results}\label{sec:lin-algebra}

For vectors $\vec{a} \in \mathbb{R}^n$, we denote the usual $\ell_p$-norm as $\|\vec{a}\|_p = (\sum_i a_i^p)^{1/p} $. We also make use of some element-wise matrix norms:

\begin{definition}[Element-wise matrix norms]
    Let $\mat{M} \in \Rbb^{n \times m}$, and let $p,q \in [1,\infty]$. We write
    \[\norm{\mat{M}}_{p,q} = \left(\sum_{j=1}^n \left(\sum_{k=1}^m \big|[\mat{M}]_{jk}\big|^p\right)^{\frac{q}{p}}\right)^{\frac1q}.\]
    We note $\norm{\cdot}_{\max} = \norm{\cdot}_{\infty,\infty}$ and $\norm{\cdot}_F = \norm{\cdot}_{2,2}$ are the max-norm and Frobenius norm, respectively.
\end{definition}

We can think of this norm intuitively as follows: first we take the vector of $\ell_p$-norms of all the rows of the matrix, and then we take the $\ell_q$-norm of the resulting vector. If we want to swap the roles of rows and columns, we can consider the norm of the transpose instead. An important subtlety is that $\norm{\mat{A}}_{p,q}$ is in general not equal to $\norm{\mat{A}^T}_{q,p}$. On the other hand, if $p = q$, then the norm for any given matrix is the same as the norm for its transpose.

Using this interpretation of the element-wise matrix norms, we also observe that H\"older's inequality holds.

\begin{lemma}[H\"older's inequality for element-wise matrix norms]\label{lem:holder}
    Let $\mat{M} \in \Rbb^{n \times m}$. Let $1 \leq p' \leq p \leq \infty$, and $1 \leq q' \leq q \leq \infty$. Then,
    \[\norm{\mat{M}}_{p,q} \leq \norm{\mat{M}}_{p',q'} \leq n^{\frac1q-\frac1{q'}}m^{\frac1p-\frac1{p'}}\norm{\mat{M}}_{p,q}.\]
\end{lemma}

As a preliminary for classical mean estimation, we recall the following inequalities. They are slight reformulations of the concentration inequalities known as the vector and matrix Bernstein inequalities, see e.g.~\cite{gross2011recovering,tropp2015introduction}.

\begin{lemma}[Vector and matrix Bernstein inequality]
    \label{lem:vector-bernstein-inequality}
    Let $\widetilde{\vec{a}}^{(L)} = \sum^{\kappa}_{i=1} \vec{a}_i/L$ be an empirical average over $L$ independent and indentically distributed $d$-dimensional random variables $\{\vec{a}_i\}_i$ which satisfy $\mathbb{E}[\vec{a}_i]=\vec{a}$. Then, with probability at least $(1-\delta)$ we have
    \[
        \big\| \widetilde{\vec{a}}^{(L)} - \vec{a} \big\|_{\infty} = O\! \left( \sqrt{\frac{\max_{i \in [d]} [\mat{\Sigma}]_{ii} \log(d/\delta)}{L}} \right)\,, \qquad \text{and} \qquad \big\| \widetilde{\vec{a}}^{(L)} - \vec{a} \big\|_2 = O\! \left( \sqrt{\frac{\Tr[\mat{\Sigma}] \log(1/\delta)}{L}} \right)\,,
    \]
    where $\mat{\Sigma} = \mathbb{E}[(\vec{a}_i - \vec{a})^T(\vec{a}_i - \vec{a})]$ is the covariance matrix.
\end{lemma}

\subsection{Quantum computational model and memory}\label{sec:computational-model}

In both our classical and quantum algorithms, we assume access to different levels of memory architecture (see ``Data model" in Table \ref{tab:results-introduction}). Specifically, some algorithms call an architecture that allows for efficient reading and writing operations -- given a register containing a memory address, we assume to be able to swap the bit at that location in memory with one from another register. In the quantum setting, this means that we assume to have access to a quantum random-access gate (QRAM gate), that for all data $\vec{x} \in \{0,1\}^d$, $j \in [d]$ and $b \in \{0,1\}$ acts as

\[\mathrm{QRAM} : \ket{x_1, \dots, x_d}\ket{j}\ket{b} \mapsto \ket{x_1,\dots,x_{j-1},b,x_{j+1},\dots,x_d}\ket{j}\ket{x_j}.\]

We remark that this is a stronger assumption than having access to a quantum read-only memory (QROM),\footnote{We also remark here that the many different nomenclatures for read-only and read-write operations on quantum memory. Instead of the acronyms QROM/QRAM~\cite{cornelissen2025quantum}, one can find QRAM/QRAG~\cite{allcock2023quantum}, QCRAM/full-QRAM~\cite{apeldoorn2020QConvexOptThesis}, QACM/QAQM~\cite{naya2020optimal}, QRACM/QRAQM \cite{kuperberg2013another} and more. We prefer the nomenclature QROM/QRAM, to mimic the nomenclature ROM/RAM for the corresponding classical counterparts.} where one assumes to have access to a gate
\[
\mathrm{QROM} : \ket{x_1, \dots, x_d}\ket{j}\ket{b}
\mapsto \ket{x_1, \dots, x_d}\ket{j}\ket{b \oplus x_j}.
\]

Having access to efficiently-addressable write operations is standard in the classical literature, but its quantum counterpart has been somewhat contentious, particularly as implementing an $O(d)$-sized QRAM/QROM fault-tolerantly may require $O(d)$ (classical) effort \cite{jaques2023qram, dalzell2025distillation}. We stress here that we only require our quantum algorithm to be able to perform the same operations as its corresponding classical counterpart, i.e., if our original classical algorithms do not actually use efficiently-addressable writing operations, then neither do our corresponding quantum algorithms. In this work we focus on theoretical comparison between the classical and quantum computational models. 

In other settings, we do not need to assume any efficient read or write operations. Instead, we operate in a \textit{circuit model}. Here, our classical (or quantum) algorithm can be written wholly in terms of a fixed classical (or quantum) circuit acting on input data \cite{sipser1996introduction}. When we consider this setting, it is sufficient to consider the input data in the form of a binary encoding $\ket{x_1} \otimes \cdots \otimes \ket{x_d}$.

So far in the above exposition we have used a simplifying picture of binary data. In both the classical and quantum models, we make the assumption that we can represent real numbers exactly in memory, and that we can retrieve them and perform basic arithmetic operations (that is, addition, multiplication, subtraction and division) on them in constant time. This is a standard assumption in the classical literature and we make the same assumption quantumly.

\subsection{Classical subroutines}

First, we recall how we can sample from a given distribution of $n$ outcomes, with probabilities written in memory. The most efficient method to achieve this is known as the alias method, and we recall a result by Vose.

\begin{theorem}[Alias method \cite{vose1991linear}]
    \label{thm:classical-sampling-from-distribution}
    Given a probability distribution over $n$ elements, with the probabilities written in memory, we can sample from it using $O(n)$ preprocessing time in the RAM-model, and $O(1)$ time per sample performing only read operations.
\end{theorem}

The fact that we only perform read operations per sample is important, because it means that we can port this algorithm to the QROM-setting. Indeed, we can implement Vose's sampling algorithm using a quantum read-only memory of size $O(n)$, which we initialize using classical RAM only, and then access with read-only operations by a quantum algorithm.
        
\subsection{Quantum subroutines}\label{sec:quantum-primitives}

We will make use of a multivariate quantum mean estimation routine from \cite{cornelissen2022multiVariateMeanEstimation} (and later refined in \cite{tang2025more}), which we restate here for convenience.

\begin{theorem}[Quantum mean estimation {\cite[Theorem 3.4]{cornelissen2022multiVariateMeanEstimation}}]\label{thm:quantum-mean-estimation}
    \label{thm:multivariate-mean-estimation}
    Let $d \in \mathbb{N}$, $\delta \in (0,1)$, $(\Omega,\Pbb)$ be a probability space, and $\vec{X} : \Omega \to \Rbb^d$ a random variable. Suppose we have access to the following two operations:
    \[U_\Pbb : \ket{0} \mapsto \sum_{\omega \in \Omega} \sqrt{\Pbb(\omega)}\ket{\omega}, \qquad \text{and} \qquad O_{\vec{X}} : \ket{\mat{Y}}\ket{\omega}\ket{0} \mapsto \ket{\mat{Y}}\ket{\omega}\ket{\mat{Y}^T\vec{X}(\omega)},\]
    where $\omega \in \Omega$ and $\vec{Y} \in \Rbb^d$. Let $\vec{\mu}$ and $\mat{\Sigma}$ be the expectation and covariance matrix of $X$, and let $N \geq \log(d/\delta)$. Then, we can produce a vector $\widetilde{\vec{\mu}} \in \Rbb^d$, such that $\norm{\widetilde{\vec{\mu}} - \vec{\mu}}_{\infty} \leq \sqrt{\Tr[\mat{\Sigma}]}\log(d/\delta)/N$, with probability at least $(1-\delta)$, with $\widetilde{O}(N)$ queries to $U_\Pbb$, $O_{\vec{X}}$ and an inner product routine between the random variable $X(\omega)$ and an arbitrary vector $\mat{Y} \in \Rbb^d$, and $\widetilde{O}(d)$ additional elementary gates.
\end{theorem}

The time complexity is not explicitly analyzed in \cite{cornelissen2022multiVariateMeanEstimation}, so we discuss this now here. The three key steps of the algorithm consist of (1) preparation of a uniform superposition of a grid of points $\mat{Y} \in G \subset (-1/2, 1/2)^d$;  (2) computing the inner product of $\mat{Y}$ with a data vector; (3) applying the inverse quantum Fourier transform.

Step (1) costs $\widetilde{O}(d)$ without any need of quantum memory assumptions, due to efficient state preparation algorithms for uniform superpositions. The inverse quantum Fourier transform is also efficient. We can thus focus on step (2) as the non-generic cost. Explicitly, we require an operation that computes the inner product, i.e., for all $\mat{Y} \in \Rbb^d$,
    \[O_{\vec{X}}' : \ket{\mat{Y}}\ket{\omega}\ket{0} \mapsto \ket{\mat{Y}}\ket{\omega}\ket{\mat{Y}^T \mat{X}(\omega)}.\]
Since we assume to be able to do arithmetic operations in unit time, this leads to an additional $O(d)$ overhead per sample in the circuit model, bringing the total time complexity to $\widetilde{O}(dN)$. A similar observation was made in \cite{apers2023quantumSpeedupsLP}. However, what we will use in this work is that if we can implement this inner product more efficiently, say with time complexity $\eta$, then the resulting total time complexity becomes $\widetilde{O}(\eta N + d)$.



\section{The Cohen-Lewis algorithm}
    
In this section, we describe our first classical and quantum algorithms for implementing approximate matrix multiplication. The idea of this algorithm is based on the seminal work by Cohen and Lewis~\cite{cohen1999approximating}.

Suppose that we are given $k$ matrices, $\mat{A}_1, \dots, \mat{A}_k$, such that $\mat{A}_j \in \Rbb^{n_{j-1},n_j}$, for all $j \in [k]$. Then, $\mat{A}_1 \cdots \mat{A}_k$ is well-defined, and our goal is to output an approximation of this matrix product. We write $n = n_0$ and $m = n_k$, so that the resulting matrix product is of dimension $n \times m$. We assume that full descriptions of the matrices are available to us in memory.

\subsection{Stochastic matrix product decomposition}

We start by showing that we can assume all matrices have absolute row-sum $1$, without loss of generality. To that end, we show that we can classically perform a stochastic matrix product decomposition as a preprocessing step, using time and memory that is similar to the size of the matrices themselves.

We start by formally introducing the decomposition.

\begin{theorem}[Stochastic matrix product decomposition]
    \label{thm:matrix-product-decomposition}
    Let $k \in \mathbb{N}$, $n_j \in \mathbb{N}$ for all $\ell \in [k]_0$, and let $\mat{A}_\ell \in \Rbb^{n_{\ell-1} \times n_\ell}$ for all $\ell \in [k]$. We recursively define diagonal matrices $\mat{D}_0, \dots, \mat{D}_k$, such that
    \[
    \mat{D}_k = I\,,
    \qquad \text{and} \qquad \forall \ell \in [k], i \in [n_{\ell-1}]\,, \quad 
    [\mat{D}_{\ell-1}]_{ii} = \sum_{j=1}^{n_{\ell}} [\ols{\mat{A}}_{\ell}]_{ij} \cdot [\mat{D}_{\ell}]_{jj}
    = \|[\ols{\mat{A}}_\ell \mat{D}_\ell]_{i\bullet}\|_1\,.
    \]
    
    Next, for all $\ell \in [k]$, we let $\mat{A}_{\ell}' = \mat{D}_{\ell-1}^{-1}\mat{A}_{\ell}\mat{D}_{\ell}$. Then, $\mat{A}_1 \cdots \mat{A}_k = \mat{D}_0\mat{A}_1' \cdots \mat{A}_k'$, all $(\mat{A}_j')$'s have all row-sums equal to $1$, and for all $\ell \in [k]$ and $i \in [n_{\ell-1}]$,
    \[
    [\mat{D}_{\ell-1}]_{ii} = \sum_{j=1}^{n_{\ell}} [\ols{\mat{A}_{\ell}} \cdots \ols{\mat{A}_k}]_{ij}
    = [\ols{\mat{A}_{\ell}} \cdots \ols{\mat{A}_k} \mathbf{1}_{n_k}]_i\,.
    \]
    Moreover, we can compute all $\mat{D}_j$'s and $(\mat{A}_j')$'s classically in time linear in the input size.
\end{theorem}

\begin{proof}
    It is clear by plugging in the definitions of $\mat{A}_\ell'$ that for all $\ell \in [k]$,
    \[\mat{D}_{\ell-1}\mat{A}_\ell' \cdots \mat{A}_k' = \mat{A}_\ell \cdots \mat{A}_k\mat{D}_k = \mat{A}_{\ell} \cdots \mat{A}_k.\]
    Moreover, for every $\ell \in [k]$ and $i \in [n_{\ell-1}]$, we have
    \[\sum_{j=1}^{n_{\ell}} |[\mat{A}_{\ell}']_{ij}| = \frac{1}{[\mat{D}_{\ell-1}]_{ii}} \cdot \sum_{j=1}^{n_{\ell}} |[\mat{A}_{\ell}]_{ij}| \cdot [\mat{D}_{\ell}]_{jj} = 1.\]
    The characterization of $[\mat{D}_{\ell-1}]_{ii}$ follows by induction. Finally, it is clear from the recursive definition how we compute the $\mat{D}_{\ell}$'s and $(\mat{A}_{\ell}')$'s in linear time.
\end{proof}

Using this decomposition, we can also compute some element-wise norms of the matrix product efficiently. Specifically, for $p = 1$ and $q \in [1,\infty]$, we show that the element-wise matrix norm can be computed in time linear in the input size.

\begin{lemma}
    \label{lem:element-wise-norm-characterization}
    Let $k \in \mathbb{N}$, $n_j \in \mathbb{N}$ for all $\ell \in [k]_0$, and let $\mat{A}_\ell \in \Rbb^{n_{\ell-1} \times n_\ell}$ for all $\ell \in [k]$. Let $q \in [1,\infty]$. Let $\mat{D}_0$ be as in \cref{thm:matrix-product-decomposition}, and let $\vec{d}_0$ be the vector containing the diagonal of $\mat{D}_0$. Then,
    \[\norm{\ols{\mat{A}_1} \cdots \ols{\mat{A}_k}}_{1,q} = \norm{\vec{d}_0}_q,\]
    and as such it can be computed in time linear in the total size of the input.
\end{lemma}

\begin{proof}
    The result follows directly from the characterization of $(\mat{D}_0)_{ii}$ in \cref{thm:matrix-product-decomposition}.
\end{proof}

We now state the crucial observation that motivates the stochastic matrix product decomposition.
It is a variant on path-integral Monte Carlo, and it mimics the approach taken by Cohen and Lewis~\cite{cohen1999approximating}, but differs in one crucial aspect: Cohen and Lewis estimate every row of the product matrix in a separate subroutine. Here, we sample over the rows as well, leading to a random variable with expectation exactly equal to the matrix product we are after. This is the core of the improvement we present in this work.

\begin{lemma}
    \label{lem:CL-expectation}
    Let $k \in \mathbb{N}$, $n_j \in \mathbb{N}$ for all $\ell \in [k]_0$, and let $\mat{A}_\ell \in \Rbb^{n_{\ell-1} \times n_\ell}$ for all $\ell \in [k]$. Let $\mat{D}_0, \mat{A}_1', \dots, \mat{A}_k'$ be the stochastic matrix product decomposition from \cref{thm:matrix-product-decomposition}. Let $\vec{d}_0$ be the vector containing the diagonal entries of $\mat{D}_0$. Then, we define the probability distribution $p$ by sampling $j_0 \sim \vec{q}$, for a probability distribution $\vec{q}$, and then $j_{\ell} \sim (\ols{\mat{A}_{\ell}'})_{j_{\ell-1},\bullet}$, for all $\ell \in [k]$. Finally, we define the random variable
    \[\mat{X}(j_0, \dots, j_k) = \frac{d_{j_0}}{q_{j_0}} \cdot \left(\prod_{\ell=1}^k \mathrm{sign}([\mat{A}_{\ell}']_{j_{\ell-1},j_\ell})\right) \vec{e}_{j_0} \vec{e}_{j_k}^T,\]
    where $\mathrm{sign}(x)$ is $1$ if $x \geq 0$, and $-1$ otherwise. Then,
    \[\mat{A}_1 \cdots \mat{A}_k = \underset{(j_0, \dots, j_k) \sim p}{\E} \left[\mat{X}\right], \qquad \text{and} \qquad \underset{(j_0, \dots, j_k) \sim p}{\Var}[\mat{X}_{ij}] \leq \frac{(\vec{d}_0)_i}{\vec{q}_i} \cdot (\mat{\ols{A}}_1 \cdots \mat{\ols{A}}_k)_{ij}.\]
    Hence, if we take $\vec{q} = \vec{d}_0/\norm{\vec{d}_0}_1$, then we obtain 
    \[\max_{i \in [n_0], j \in [n_k]} \Var[\mat{X}_{ij}] \leq \|\mat{\ols{A}}_1 \cdots \mat{\ols{A}}_k\|_{\infty,\infty} \cdot \|\mat{\ols{A}}_1 \cdots \mat{\ols{A}}_k\|_{1,1}, \qquad \text{and} \qquad  \Tr[\mat{\Sigma}] \leq \|\mat{\ols{A}}_1 \cdots \mat{\ols{A}}_k\|_{1,1}^2,\]
    where $\mat{\Sigma} = \Cov(\mat{X})$. If we instead take ${q}_i \propto \|(\mat{\ols{A}}_1 \cdots \mat{\ols{A}}_k)_{i,\bullet}\|_1 \cdot \|(\mat{\ols{A}}_1 \cdots \mat{\ols{A}}_k)_{i,\bullet}\|_{\infty}$, then
    \[\max_{i \in [n_0], j \in [n_k]} \Var[\mat{X}_{ij}] \leq \sum_{i=1}^{n_0} \|(\mat{\ols{A}}_1 \cdots \mat{\ols{A}}_k)_{i,\bullet}\|_1 \cdot \|(\mat{\ols{A}}_1 \cdots \mat{\ols{A}}_k)_{i,\bullet}\|_{\infty},\]
    and if we take $\vec{q} = \vec{d}_0^2 / \norm{\vec{d}_0}_2^2$, then we obtain
    \[\max_{i \in [n_0], j \in [n_k]} \Var[\mat{X}_{ij}] \leq \norm{\mat{\ols{A}}_1 \cdots \mat{\ols{A}}_k}_{1,2}^2.\]
\end{lemma}

\begin{proof}
Unwrapping the expectation yields
    \begin{align*}
        &\underset{(j_0, \dots, j_k) \sim p}{\E}[\mat{X}(j_0, \dots, j_k)] \\
        &\quad = \frac{(\vec{d}_0)_{j_0}}{\vec{q}_{j_0}} \cdot \sum_{j_0 = 1}^{n_0} \vec{q}_{j_0} \cdot \sum_{j_1 = 1}^{n_1} |\mat{A}_1'|_{j_0,j_1} \mathrm{sign}((\mat{A}_1')_{j_0,j_1}) \cdots \sum_{j_k=1}^{n_k} |\mat{A}_k'|_{j_{k-1},j_k}\mathrm{sign}((\mat{A}_k')_{j_{k-1},j_k}) \vec{e}_{j_0}\vec{e}_{j_k}^T = \\
        &\quad = \sum_{j_0 = 1}^{n_0} \vec{e}_{j_0} (\vec{d}_0)_{j_0}  \cdot \sum_{j_1 = 1}^{n_1} (\mat{A}_1')_{j_0,j_1} \cdots \sum_{j_k=1}^{n_k} (\mat{A}_k')_{j_{k-1},j_k}\vec{e}_{j_k}^T \\
        &\quad = \sum_{j_0 = 1}^{n_0} \vec{e}_{j_0} \vec{e}_{j_0}^T \mat{D}_0 \vec{e}_{j_0}  \cdot \sum_{j_1 = 1}^{n_1} \vec{e}_{j_0}^T \mat{A}_1' \vec{e}_{j_1} \cdots \sum_{j_k=1}^{n_k} \vec{e}_{j_{k-1}}^T\mat{A}_k'\vec{e}_{j_k}\vec{e}_{j_k}^T \\
        &\quad = \sum_{j_0' = 1}^{n_0} \vec{e}_{j_0'} \vec{e}_{j_0'}^T \mat{D}_0 \sum_{j_0 = 1}^{n_0} \vec{e}_{j_0}\vec{e}_{j_0}^T \mat{A}_1' \sum_{j_1 = 1}^{n_1}\vec{e}_{j_1}\vec{e}_{j_1}^T \cdots \mat{A}_k' \sum_{j_k=1}^{n_k}\vec{e}_{j_k}\vec{e}_{j_k}^T = \mat{D}_0\mat{A}_1' \cdots \mat{A}_k' = \mat{A}_1 \cdots \mat{A}_k.
    \end{align*}
    Next, for the variance of the $(i,j)$-th entry of $\mat{X}$, we observe that
    \begin{align*}
        \underset{(j_0, \dots, j_k) \sim p}{\Var}[\mat{X}(j_0, \dots, j_k)_{ij}] &\leq \underset{(j_0, \dots, j_k) \sim p}{\E}[\mat{X}(j_0, \dots, j_k)_{ij}^2] = \left(\frac{(\vec{d}_0)_i}{\vec{q}_i}\right)^2 \cdot \underset{(j_0, \dots, j_k) \sim p}{\Pbb} \left[j_0 = i \land j_k = j\right] \\
        &= \frac{(\vec{d}_0)_i}{\vec{q}_i} \cdot (\mat{\ols{A}}_1 \cdots \mat{\ols{A}}_k)_{ij}.
    \end{align*}
    The final expressions follow directly.
\end{proof}

From the above analysis, we observe that when we choose $\vec{q} \propto \vec{d}_0$, we minimize the trace of the covariance matrix, which means that this choice yields the best algorithm for the Frobenius-norm problem. On the other hand, by choosing $\vec{q} \propto \vec{d}_0^2$, we recover the complexity from \cite{cohen1999approximating} for the $\max$-norm problem. Both choices for $\vec{q}$ yield incomparable complexities for the $\max$-norm problem.

It is tempting to simply choose the optimal $\vec{q}$ for the $\max$-norm problem, which is when $q_i \propto \norm{(\mat{\ols{A}}_1 \cdots \mat{\ols{A}}_k)_{i,\bullet}}_{\infty} \cdot \norm{(\mat{\ols{A}}_1 \cdots \mat{\ols{A}}_k)_{i,\bullet}}_1$. However, the problem is that it is not clear how to compute this choice for $\vec{q}$ efficiently in linear time as a preprocessing step. Hence, we are left with an algorithm with incomparable complexity to the max-norm algorithm given by Cohen and Lewis' approach, for instance if we take $\vec{q} = \vec{d}_0/\norm{\vec{d}_0}_1$ (see next section for further discussion). If we additionally have good upper bounds on $\norm{(\mat{\ols{A}}_1 \cdots \mat{\ols{A}}_k)_{i,\bullet}}_{\infty}$ a priori, for all $i \in [n]$, then we could indeed  further improve the algorithm.

\subsection{Classical algorithm}

The idea for the classical algorithm is to compute the expectation from \cref{lem:CL-expectation} using Monte Carlo sampling. It remains to figure out the number of samples required, and the cost to perform sampling, which we analyze in the following theorem.

\begin{theorem}[Classical random walk algorithm]\label{thm:improved-cohen-lewis}
    Let $k \in \mathbb{N}$, $n_j \in \mathbb{N}$ for all $\ell \in [k]_0$, and let $\mat{A}_\ell \in \Rbb^{n_{\ell-1} \times n_\ell}$ for all $\ell \in [k]$. Let $\varepsilon > 0$ and $\delta \in (0,1)$, and let $n_0 = n$ and $n_k = m$. Then, with classical preprocessing time linear in the size of the input up to polylogarithmic factors, we give a classical algorithm in the RAM model with run-time
    \[\widetilde{O}\left(k \cdot \frac{\norm{\mat{\ols{A}}_1 \cdots \mat{\ols{A}}_k}_{\infty,\infty} \cdot \norm{\mat{\ols{A}}_1 \cdots \mat{\ols{A}}_k}_{1,1}}{\varepsilon^2} \cdot \log\left(\frac{d}{\delta}\right)\right) \quad \text{and} \quad \widetilde{O}\left(k \cdot \frac{\norm{\ols{\mat{A}_1} \cdots \ols{\mat{A}_k}}_{1,1}^2}{\varepsilon^2} \cdot \log\left(\frac{1}{\delta}\right)\right),\]
    that computes an approximation $\widetilde{\mat{A}}$ of $\mat{A}_1 \cdots \mat{A}_k$ with probability at least $(1-\delta)$, satisfying the precision bounds $\norm{\widetilde{\mat{A}} - \mat{A}_1 \cdots \mat{A}_k}_{\max} \leq \varepsilon$ and $\norm{\widetilde{\mat{A}} - \mat{A}_1 \cdots \mat{A}_k}_F \leq \varepsilon$, respectively.
\end{theorem}

\begin{proof}
    We take the average of $N$ realizations of the random variable from \cref{lem:CL-expectation}. Using the bound on the trace of the covariance matrix, and the vector-Bernstein inequality~\cref{lem:vector-bernstein-inequality}, we obtain that with probability at least $(1-\delta)$,
    \[\norm{\widetilde{\mat{A}} - \mat{A}_1 \cdots \mat{A}_k}_F = \mathcal{O}\left(\norm{\ols{\mat{A}_1} \cdots \ols{\mat{A}_k}}_{1,1} \cdot \sqrt{\frac{\log(1/\delta)}{N}}\right),\]
    from which we deduce that we can take $N = \Theta(\norm{\ols{\mat{A}_1} \cdots \ols{\mat{A}_k}}_{1,1}^2 / \varepsilon^2 \cdot \log(1/\delta))$ to ensure that the Frobenius error is at most $\varepsilon$. A similar analysis yields the corresponding complexity for the $\max$-norm case.
    
    It remains to prove that we can generate each sample in time $O(k)$. To that end, observe from \cref{thm:matrix-product-decomposition} that we can compute the stochastic matrix product decomposition in time linear in the size of the input, and from \cref{lem:element-wise-norm-characterization} that we can compute the choice for $N$ in time linear in the size of the input as well. Next, we observe that we can use \cref{thm:classical-sampling-from-distribution} to sample from $\vec{d}_0/\norm{\vec{d}_0}_1$ and from any of the rows of the matrices $\mat{A}_{\ell}'$ in polylogarithmic time, with total preprocessing time again linear in the size of the input. Finally, since we need to sample from $k$ probability distributions, and multiply $k$ signs to obtain one realization of the random variable, the cost per sample is indeed $O(k)$.
\end{proof}

The resulting complexity for the Frobenius norm unconditionally improves over the result in \cite{cohen1999approximating}. Indeed, by comparing the expressions in \cref{tbl:complexities}, we can use H\"older's inequality to conclude that
\[\|\mat{A}_1 \cdots \mat{A}_k\|_{1,1}^2 \leq n\|\mat{A}_1 \cdots \mat{A}_k\|_{1,2}^2.\]
The situation is a bit more complicated for the $\max$-norm problem, since the bound from \cite{cohen1999approximating} and our complexity are incomparable. Indeed, if the resulting matrix product $\mat{A}$ is the $n \times n$ all-ones matrix, then the Cohen-Lewis bound is $\norm{\mat{A}}_{1,2}^2 = n^3$, whereas our bound is $\norm{\mat{A}}_{\infty,\infty} \cdot \norm{\mat{A}}_{1,1} = n^2$. On the other hand, if the resulting matrix product $\mat{A}$ is $\mathrm{diag}(\sqrt{n}, 1, \dots, 1)$, then the Cohen-Lewis bound is $2n$, whereas our bound is $n^{3/2}$.

\subsection{Quantum algorithm}

In the quantum setting, we wish to speed-up the classical algorithm using the multivariate quantum mean estimation routine~\cite{cornelissen2022multiVariateMeanEstimation}. We describe our result in the following theorem, in the proof of which we compute the costs of implementing the input routines for this procedure.

\begin{theorem}[Quantum random walk algorithm]\label{thm:quantum-cohen-lewis}
    Let $k \in \mathbb{N}$, $n_j \in \mathbb{N}$ for all $\ell \in [k]_0$, and let $\mat{A}_\ell \in \Rbb^{n_{\ell-1} \times n_\ell}$ for all $\ell \in [k]$. Let $\varepsilon > 0$ and $\delta \in (0,1)$, and let $n_0 = n$ and $n_k = m$. Then, with classical preprocessing in the RAM model with time linear in the size of the input up to polylogarithmic factors, there exists a quantum algorithm in the QRAM model with run-time
    \[\widetilde{O}\left(k \cdot \frac{\norm{\ols{\vec{A}_1} \cdots \ols{\vec{A}_k}}_{1,1}}{\varepsilon} \cdot \log\left(\frac{nm}{\delta}\right)\right),\]
    that computes an approximation $\widetilde{\mat{A}}$ satisfying $\norm{\widetilde{\mat{A}} - \mat{A}_1 \cdots \mat{A}_k}_{\max} < \varepsilon$, with probability at least $(1-\delta)$.
\end{theorem}

\begin{proof}
    We take $N$ as the number of q-samples in the quantum multivariate mean estimation algorithm. Then, combining \cref{lem:CL-expectation} and \cref{thm:multivariate-mean-estimation}, the resulting estimate $\widetilde{\mat{A}}$ satisfies with probability at least $(1-\delta)$,
    \[\norm{\widetilde{\mat{A}} - \mat{A}_1 \cdots \mat{A}_k}_{\max} = \widetilde{O}\left(\norm{\ols{\mat{A}_1} \cdots \ols{\mat{A}_k}}_{1,1} \cdot \frac{\log(nm/\delta)}{N}\right),\]
    and so it suffices to choose $N = \Theta(\norm{\ols{\mat{A}_1} \cdots \ols{\mat{A}_k}}_{1,1} \log(nm/\delta)/\varepsilon)$ to ensure that the error is at most $\varepsilon$.
    
    It remains to prove that we can implement the input oracles in time $\widetilde{O}(k)$. To that end, observe that coherent sampling from $p$ can be done by implementing the classical sampling producedure reversibly, as described in \cite{montanaro2015QMonteCarlo}, with the same time complexity $\widetilde{O}(k)$. Thus, total cost per sample for implementing $U_{\Pbb}$ is indeed $\widetilde{O}(k)$.
    
    Finally we compute the cost of implementing the oracle computing the random variable. To that end, observe that it must implement the following operation:
    \[O_{\mat{X}} : \ket{j_0,\dots,j_k}\ket{0} \mapsto \ket{j_0,\dots,j_k}\ket{\prod_{\ell=1}^k \mathrm{sign}([\mat{A}_{\ell}']_{j_{\ell-1},j_{\ell}}) \cdot \vec{e}_{j_0} \vec{e}_{j_k}^T}.\]
    This essentially boils down to $k$ look-ups of the signs of the entries in the input, and then $O(k)$ multiplications of signs. This can be implemented in $\widetilde{O}(k)$ time with the aid of QROM access to preprocessed data. We complete the proof by observing that we can also implement the inner product routine in the same way if we have access to QRAM, i.e., we can store the matrix $\ket{\mat{Y}}$ in QRAM and index into it efficiently to spend at most $\widetilde{O}(k)$ time to implement
    \[O_{\mat{X}}' : \ket{\mat{Y}}\ket{j_0, \dots j_k}\ket{0} \mapsto \ket{\mat{Y}}\ket{j_0, \dots, j_k}\ket{\prod_{\ell=1}^k \mathrm{sign}([\mat{A}_{\ell}']_{j_{\ell-1},j_{\ell}}) Y_{j_0,j_k}}.\qedhere\]
\end{proof}

This approach gives us a quantum algorithm that estimates the matrix product in $\max$-norm, whereas the corresponding classical algorithm provides an approximation in terms of Frobenius norm. We can convert between the two using norm conversion, since $\norm{\cdot}_F \leq \sqrt{nm}\norm{\cdot}_{\max}$. The resulting expressions can be found in \cref{tbl:complexities} and Table \ref{tab:results-full}.

\section{Sketching-based algorithms}

\subsection{Generic analysis}

In this section we provide a unifying picture for matrix multiplication algorithms via sketching, which encapsulates the algorithms of \cite{drineas2006fast} and \cite{sarlos2006improved}, and also enables us to construct classical algorithms for multiple matrices as well as quantum algorithms. For simplicity, in this part of the exposition let us assume both $\mat{A},\mat{B} \in \mathbb{R}^{n \times n}$, though the discussion also generalizes to rectangular matrices. Our goal is to construct a matrix-valued random variable $\mat{S} \in \mathbb{R}^{c \times n}$ (with $c < n$) such that for all $\mat{A},\mat{B}$ we have
\begin{equation}\label{eq:sketch-exp-val}
    \E[\mat{A}\mat{S}^T \mat{S}\mat{B}] = \mat{AB}\,,
\end{equation}
which in particular implies that $\E[\mat{S}^T \mat{S}] = \mat{I}$.
The matrices $\mat{A}\mat{S}^T$ and $\mat{S}\mat{B}$ may be thought of as sketches of $\mat{A}$ and $\mat{B}$ respectively. Specifically, judicious choices of $\mat{S}$ can enable efficient approximation of $\mat{AB}$, as the complexity of evaluating $\mat{A}\mat{S}^T \mat{S}\mat{B}$ is $O(n^2c^{\omega -2})$, compared to $O(n^{\omega})$ for the exact matrix product (here $\omega$ can either be $\omega = 3$, if we disallow fast matrix multiplication, or else it is the matrix multiplication exponent $\omega < 2.37$). Supposing the rows of $\mat{S}$ are chosen independently, and we denote them as the column vectors $\{\hat{\vec{s}}_j\}_{j=1}^c$, then we can also write

\begin{equation}\label{eq:sketch-exp-val-2}
    \E[\mat{A}\vec{s} \vec{s}^T\mat{B}] = \mat{AB} \,,
\end{equation}
for a vector-valued random variable {$\vec{s}:=\hat{\vec{s}}\sqrt{c}$}, and we may view the matrix multiplication as a mean over outer products, each costing $O(n^2)$ times to evaluate.  
In this perspective, we can view one instance of $\mat{A}\mat{S}^T \mat{S}\mat{B}$ as a \textit{sample mean} of $\mat{A}\vec{s} \vec{s}^T\mat{B}$ over $c$ samples. For the sake of analysis, we use this picture and find it useful. However, we stress that the quantum and classical algorithm algorithmically adopt different approaches. The classical algorithms use Eq.~\eqref{eq:sketch-exp-val}, explicitly computing $\mat{A}\mat{S}^T \mat{S}\mat{B}$, and thus can exploit fast matrix multiplication; whereas the quantum algorithms use Eq.~\eqref{eq:sketch-exp-val-2} and explicitly implement a mean estimation of outer products.

Now that the problem is framed in terms of mean estimation, we can determine the convergence speed by investigating properties of the covariance matrix so that we may use Lemma \ref{lem:vector-bernstein-inequality} and Theorem \ref{thm:quantum-mean-estimation} to analyze the classical and quantum algorithm respectively. 

\begin{lemma}[Covariance matrix for sketching algorithms]\label{lem:covariance-sketch}
Consider the matrix-valued random variable $\mat{A}\vec{s} \vec{s}^T\mat{B}$. The covariance matrix $\mat{\Sigma}$ satisfies 
\begin{align*}
    \max_i[\mat{\Sigma}]_{ii} &=  \max_{ij}\Var\big[[\mat{A}\vec{s} \vec{s}^T\mat{B}]_{ij}\big] = \max_{ij}\sum_{\substack{k\ell qr }} [\mat{A}]_{ik} [\mat{B}]_{\ell j} [\mat{A}]_{iq} [\mat{B}]_{rj} \Cov\big[ s_{k}s_{\ell}, s_{q}s_{r} \big] \,,\\
    \Tr[\mat{\Sigma}] &= \E\! \left[ \| \mat{AB} - \mat{A}\mat{S}^T \mat{S}\mat{B} \|_F^2 \right] =  \sum_{ij}\sum_{\substack{k\ell qr }} [\mat{A}]_{ik} [\mat{B}]_{\ell j} [\mat{A}]_{iq} [\mat{B}]_{rj} \Cov\big[ s_{k}s_{\ell}, s_{q}s_{r} \big] \,.
\end{align*}
\end{lemma}

\begin{proof}
    In order to define the covariance matrix in terms of a vector-valued random variable, denote the super index $\alpha = (i,j)$. Recall that the diagonal entries of the covariance matrix are variances, that is
    \begin{equation*}
        [\mat{\Sigma}]_{\alpha \alpha} = \Var\left[[\mat{A}\vec{s} \vec{s}^T\mat{B}]_{\alpha}\right] = \Var\left[[\mat{A}\vec{s} \vec{s}^T\mat{B}]_{ij}\right]
    \end{equation*}
    From this the first expression for $\max_{\alpha} [\mat{\Sigma}]_{\alpha \alpha}$ follows, and the first expression for $\Tr[\mat{\Sigma}]$ follows by nothing that 
    \begin{align}
        \Tr[\mat{\Sigma}] &=  \sum_{ij}\Var\big[[\mat{A}\vec{s} \vec{s}^T\mat{B}]_{ij}\big] \nonumber \\
        &= \sum_{ij} \E\! \left[  (\mat{A}\vec{s} \vec{s}^T\mat{B} -  \mat{AB} )^2_{ij} \right] \nonumber \\
        &=  \E\! \left[\sum_{ij}   (\mat{A}\vec{s} \vec{s}^T\mat{B} -  \mat{AB} )^2_{ij} \right] \nonumber \\
        &= \E\! \left[ \| \mat{AB} - \mat{A}\vec{s} \vec{s}^T\mat{B} \|_F^2 \right] \,.\nonumber
    \end{align}
    The final expressions can be gleaned by using the standard formula for the variance of a sum of random variables
    \begin{align*}
        \Var\big[[\mat{A}\vec{s} \vec{s}^T\mat{B}]_{ij}\big] &=  \Var\Big[ \sum_{\substack{k\ell qr }} [\mat{A}]_{ik} [\vec{s} \vec{s}^T]_{k\ell} [\mat{B}]_{\ell j}  [\mat{A}]_{iq} [\vec{s} \vec{s}^T]_{qr} [\mat{B}]_{rj} \Big]  \\ &= \sum_{\substack{k\ell qr }} [\mat{A}]_{ik} [\mat{B}]_{\ell j} [\mat{A}]_{iq} [\mat{B}]_{rj} \Cov\left[[\vec{s} \vec{s}^T]_{k\ell}, [\vec{s} \vec{s}^T]_{qr}\right]\,,
    \end{align*}
    where to obtain the final expression we observe that $[\vec{s} \vec{s}^T]_{k\ell} = s_k s_{\ell}$.
\end{proof}

Lemma \ref{lem:covariance-sketch} shows that the convergence of both classical and quantum algorithms is controlled simply by covariances of elements of $\vec{s}$. In what follows we present two choices for $\vec{s}$, and their implications for both classical and quantum algorithms. 

\subsection{Classical algorithms}

Let us now consider a general rectangular matrix product $\mat{AB}$ for $\mat{A} \in \mathbb{R}^{n \times q}$, $\mat{B} \in \mathbb{R}^{q \times m}$. In classical algorithms we will sample one matrix $\mat{S} \in \mathbb{R}^{c \times q}$ with rows constructed independently from some distribution, and explicitly compute $\mat{A} \mat{S}^T \mat{S} \mat{B}$ (recall that in order to analyze convergence, we will view this as taking a sample mean over $c$ samples of $\mat{A}\vec{s} \vec{s}^T\mat{B}$). This costs $O((mn + nq + mq) c^{\omega - 2})$ runtime. We recall that the approximation error in Frobenius norm and max-norm are controlled by $\Tr[\mat{\Sigma}]$ and $\max_{i} [\mat{\Sigma}]_{ii}$ respectively.

\subsubsection{Tug-of-war sketch}\label{sec:tug-of-war}

In the tug-of-war sketch \cite{alon1996spaceComplexityFrequencyMoments} (also known as AMS sketch), $\mat{S}$ is chosen to be the dense matrix with entries $[\mat{S}]_{ij} = h_i(j)/\sqrt{c}$ with each $h_i:[n]\rightarrow \{-1, 1 \}$ drawn uniformly at random.\footnote{In fact, drawing entries randomly from a 4-wise independent hash family suffices \cite{sarlos2006improved}, though we will not exploit this.}. In the mean estimation picture, our sketching vectors $\vec{s}$ thus satisfy $s_j =h(j)$.

\begin{lemma}[Tug-of-war sketch covariance matrix]\label{lem:tug-of-war-sketch} For the tug-of-war sketch, the covariance matrix $\mat{\Sigma}$ satisfies
\[ \Tr[\mat{\Sigma}] =  \|\mat{A}\|^2_F \|\mat{B}\|_F^2 + \|\mat{AB}\|^2_F - 2\sum^n_{k=1} \|[\mat{A}]_{\bullet k}\|_2^2 \|[\mat{B}]_{k\bullet}\|_2^2  \,,\]
and 
\[ 
\max_{i} [\mat{\Sigma}]_{ii} =  \|\mat{A}\|^2_{2,\infty} \|\mat{B}^T\|^2_{2,\infty} + \|\mat{AB}\|_{\max}^2 - 2\min_{ij}\sum_{k} [\mat{A}]_{ik}^2 [\mat{B}]_{k j}^2 \,.
\]
\end{lemma}

\begin{proof} 
    We can first verify that
    \[ \E\big[[\vec{s} \vec{s}^T]_{ij}\big] =  \E [s_{i}s_{j}] = \E[ h(i) h(j)] = \delta_{ij} \quad\implies\quad \E[\vec{s} \vec{s}^T] =\mat{I}\,. \]
    The variance of diagonal terms satisfy
    \begin{align*}
        \Var\!\big[s_i^2\big] &= \E\big[s_i^4 \big] - 1 \\
        &= \E\left[ h(i)^4 \right] - 1\\
        &= 0\,.
    \end{align*}
    The variance of off-diagonal terms ($i\neq j$) satisfy
    \begin{align*}
        \Var\!\big[s_is_j\big] &= \E[(s_is_j)^2] \\
        &= \E\left[h(i)^2 h(j)^2 \right] \\
        &= 1 \,.
    \end{align*}
    Thus, we can write $\Var\left[s_is_j\right] = (1- \delta_{ij})$.
    Finally, the covariance (for $\{i,j\} \neq \{k,\ell\})$ can be evaluated as
    \begin{align*}
        \Cov\!\big[s_is_j, s_{k}s_{\ell}\big] &= \E[s_is_j s_{k}s_{\ell}] - \E\big[s_is_j\big]\, \E[s_{k}s_{\ell}] \\
        &= \E\Big[  h(i) h(j) h(k) h(\ell)  \Big] - \delta_{ij} \delta_{k \ell} \\
        &= \delta_{ij} \delta_{k \ell} - \delta_{ij} \delta_{k \ell}\\
        &= 0 \,.
    \end{align*}
    This means that overall, we have 
    \begin{equation}
        \Cov\!\big[s_is_j, s_{k}s_{\ell}\big] = (\delta_{ik}\delta_{j \ell} + \delta_{i\ell} \delta_{jk}) (1- \delta_{ij})\,. 
    \end{equation}
    We can substitute this into the expressions in Lemma \ref{lem:covariance-sketch} to give
    \begin{align*} 
        \Tr[\mat{\Sigma}]
        &= \sum_{ij} \left(\sum_{k\ell}[\mat{A}]_{ik}^2 [\mat{B}]_{\ell j}^2 + [\mat{AB}]_{ij}^2 - 2\sum_{k\ell}[\mat{A}]_{ik}^2 [\mat{B}]_{k j}^2 \right)\\
        &=   \|\mat{A}\|^2_F \|\mat{B}\|_F^2 + \|\mat{AB}\|^2_F - 2\sum_{ijk} [\mat{A}]_{ik}^2 [\mat{B}]_{kj}^2  \,,
    \end{align*}
    and
    \begin{align} 
        \max_{i} [\mat{\Sigma}]_{ii} &=  \max_{ij} \left(\sum_{k\ell}[\mat{A}]_{ik}^2 [\mat{B}]_{\ell j}^2 + [\mat{AB}]_{ij}^2 - 2\sum_{k\ell}[\mat{A}]_{ik}^2 [\mat{B}]_{k j}^2 \right) \\
        &=  \|\mat{A}\|^2_{2,\infty} \|\mat{B}^T\|^2_{2,\infty} + \|\mat{AB}\|_{\max}^2 - 2\min_{ij}\sum_{k} [\mat{A}]_{ik}^2 [\mat{B}]_{k j}^2 \,, 
    \end{align}
    as required.
\end{proof}

With this, we can write down complexities for approximate matrix multiplication using the tug-of-war sketch. This was detailed for Frobenius norm approximation in \cite{sarlos2006improved}, and we use the above to write the result for max-norm approximation also.

\begin{theorem}[Classical tug-of-war algorithm, 2 matrices]\label{thm:sarlos}
    Let $\mat{A} \in \mathbb{R}^{n \times q}$, $\mat{B} \in \mathbb{R}^{q \times m}$. By sampling a single tug-of-war matrix of appropriate dimension, we can return a $\mat{C}\in \mathbb{R}^{n \times m}$ in the circuit model with success probability at least $(1-\delta)$, which satisfies 
    
    \begin{equation}
        \|\mat{C} - \mat{AB} \|_{\max} \leq \varepsilon \quad \text{in time}\quad O\!
        \left( (mn + nq + qm) \left(\frac{\norm{\mat{A}}_{2,\infty}^2\norm{\mat{B}^T}_{2,\infty}^2}{\varepsilon^2} \right)^{\omega -2} \log \left( \frac{1}{\delta} \right) \right) \,,
    \end{equation}
    \begin{equation}
        \|\mat{C} - \mat{AB} \|_{F} \leq \varepsilon \quad \text{in time}\quad O\!
        \left( (mn + nq + qm) \left( \frac{\norm{\mat{A}}_{F}^2\norm{\mat{B}}_{F}^2}{\varepsilon^2} \right)^{\omega - 2} \log \left( \frac{1}{\delta} \right) \right)\,.
    \end{equation}
\end{theorem}

\begin{proof}
    We can directly use the matrix Bernstein inequalities of Lemma \ref{lem:vector-bernstein-inequality}, with expressions for covariance properties given by Lemma \ref{lem:tug-of-war-sketch}.
    The max-norm algorithm follows by choice of $c=L = 2\norm{\mat{A}}_{2,\infty}^2\norm{\mat{B}^T}_{2,\infty}^2\log(nm/\delta)/\varepsilon^2$ for the max-norm algorithm and $c=L = 2\norm{\mat{A}}_{F}^2\norm{\mat{B}}_{F}^2\log(1/\delta)/\varepsilon^2$ for the Frobenius norm algorithm. Finally, as discussed at the top of this section, the time complexity of computing $\mat{A} \mat{S}^T \mat{S} \mat{B}$  classically is $O((mn + nq + mq) c^{\omega - 2})$.
\end{proof}


\subsubsection{Column-sample sketch}

In the algorithm of Drineas, Kannan and Mahoney, for a matrix product of $\mat{A}\in \mathbb{R}^{n \times q}$ with $\mat{B}\in \mathbb{R}^{q \times m}$, $\mat{S} \in \mathbb{R}^{c\times q}$ is chosen to be a row $1$-sparse matrix with entries populated as $[\mat{S}]_{i \bullet} = \vec{e}^T_j/\sqrt{cp_j}$ with probability $p_j$, where we denote $\vec{e}_j$ as the unit column vector with non-zero $j$th entry, and for some probability distribution $\{p_k \}_k$ \cite{drineas2006fast}. This is their general framework, and in the end $\{p_k \}_k$ is judiciously chosen to minimize the convergence time for their algorithm in Frobenius norm approximation -- note, however, that the minimizing distribution for the Frobenius norm does not optimize convergence in other norms, and indeed we find this not to be the case. In the mean estimation picture, we can write sketching vectors chosen as $\vec{s} = \vec{e}_j/\sqrt{p_j} \in \mathbb{R}^q$ with probability $p_j$.

\begin{lemma}[Column-sample sketch covariance]\label{lem:covariance-dkm}
For the column-sample sketch as described above we have $\Cov[s_i^2, s_j^2] = \frac{1}{p_i}\delta_{ij} - 1$ and $\Cov[s_is_j, s_ku_{\ell}]=0$ for and $i \neq j$ or $k\neq \ell$. This implies that  

\[ \Tr[\mat{\Sigma}] = \sum_{i=1}^n \frac{1}{p_i} \|[\mat{A}]_{\bullet, i}\|_2^2 \|[\mat{B}]_{i, \bullet}\|_2^2  - \|\mat{AB}\|_F^2 \,,\]

\[ \max_{i} [\mat{\Sigma}]_{ii} = \max_{ij} \left( \sum_{k} \frac{1}{p_k}  [\mat{A}]_{ik}^2 [\mat{B}]_{k j}^2 - (\mat{AB})^2_{i j}  \right) \,.\]

\end{lemma}

\begin{proof}
    We can first verify that 
    \[ \E[\vec{s} \vec{s}^T]  = \sum_{j=1}^n  p_j \frac{\vec{e}_j \vec{e}^T_j}{p_j} =\mat{I}\,. \]
    We note that $\vec{s} \vec{s}^T$ is always diagonal, and thus $\Var\!\big[s_is_j\big] = \Var\!\left[[\vec{s} \vec{s}^T]_{ij}\right] = 0 $ for all $i \neq j$, and similarly for the covariance. 
    The covariances of the diagonal terms are
    \begin{align*}
        \Cov[s_i^2, s_j^2] &= \E[s_i^2 s_j^2] - \E[s_i^2]\, \E[s_j^2] \\
        &= \sum_{k=1}^q  p_k \frac{[\vec{e}_k \vec{e}^T_k]_{ii}}{p_k} \frac{[\vec{e}_k \vec{e}^T_k]_{jj}}{p_k} - \sum_{k=1}^q  p_k \frac{[\vec{e}_{k} \vec{e}^T_{k}]_{ii}}{p_k} \sum_{\ell=1}^q  p_{\ell} \frac{[\vec{e}_{\ell} \vec{e}^T_{\ell}]_{jj}}{p_{\ell}} \\
        &= \sum_{k=1}^q \frac{1}{p_k} \delta_{ik} \delta_{jk} - \sum_{k=1}^q \delta_{ik} \sum_{\ell=1}^q \delta_{\ell j} \\
        &= \frac{1}{p_i}\delta_{ij} - 1\,.
    \end{align*}
    We can now substitute this result into the expressions in Lemma \ref{lem:covariance-sketch}  and write
    \begin{align*}
    \Tr[\mat{\Sigma}]&= \sum_{ij} \left(  \sum_{k} \frac{1}{p_k} [\mat{A}]_{ik}^2 [\mat{B}]_{k j}^2 - \sum_{k \ell} [\mat{A}]_{ik} [\mat{B}]_{k j} [\mat{A}]_{i\ell} [\mat{B}]_{\ell j} \right) \\
    &=  \sum_{k=1}^n \frac{1}{p_k} \|[\mat{A}]_{\bullet, k}\|_2^2 \|[\mat{B}]_{k, \bullet}\|_2^2  - \|\mat{AB}\|_F^2 \,,
    \end{align*}
    and 
    \begin{align*}
        \max_{i} [\mat{\Sigma}]_{ii} &=  \max_{ij} \left( \sum_{k} \frac{1}{p_k}  [\mat{A}]_{ik}^2 [\mat{B}]_{k j}^2 - \sum_{k \ell}  [\mat{A}]_{ik} [\mat{B}]_{k j}  [\mat{A}]_{i \ell} [\mat{B}]_{\ell j}   \right) \\
        &= \max_{ij} \left( \sum_{k} \frac{1}{p_k}  [\mat{A}]_{ik}^2 [\mat{B}]_{k j}^2 - (\mat{AB})^2_{i j}  \right) \,.
    \end{align*}
    which completes the proof.
\end{proof}

The probability distribution can now be optimized -- note that the optimal probability distribution for the Frobenius norm (as found in \cite{drineas2006fast}) may differ from the optimal distribution for the max-norm, and indeed we find this to be the case. 

\begin{theorem}[Classical column-sampling algorithm, 2 matrices]
    Let $\mat{A} \in \mathbb{R}^{n \times q}$, $\mat{B} \in \mathbb{R}^{q \times m}$. By sampling a single column-sampling matrix of appropriate dimension, we can return a matrix $\mat{C}\in \mathbb{R}^{n \times m}$, with success probability at least $(1-\delta)$, which satisfies 
    
    \begin{equation}
        \|\mat{C} - \mat{AB} \|_{\max} \leq \varepsilon \quad \text{in time}\quad O\!
        \left( nm \left(\frac{\norm{\mat{A}^T}_{\infty,2}^2\norm{\mat{B}}_{\infty,2}^2}{\varepsilon^2} \right)^{\omega -2} \log\left( \frac{1}{\delta}\right) \right) \,,
    \end{equation}
    \begin{equation}
        \|\mat{C} - \mat{AB} \|_{F} \leq \varepsilon \quad \text{in time}\quad O\!
        \left( nm \left( \frac{\norm{\mat{A}}_{F}^2\norm{\mat{B}}_{F}^2}{\varepsilon^2} \right)^{\omega - 2} \log\left( \frac{1}{\delta}\right) \right)\,,
    \end{equation}
    where $\omega$ is the matrix multiplication exponent.
\end{theorem}

\begin{proof}

Let us consider the Frobenius-norm algorithm first. Drineas et al.~show that the minimizing probability distribution for $\Tr[\mat{\Sigma}]$ is the one with $p_k \propto \|[\mat{A}]_{\bullet, k}\|_2 \|[\mat{B}]_{k, \bullet}\|_2$, which gives

\[ \Tr[\mat{\Sigma}] =  \big(\sum_i  \|[\mat{A}]_{\bullet, i}\|_2 \|[\mat{B}]_{i, \bullet}\|_2\big)^2  - \|\mat{AB}\|_F^2  \,.\]

The Frobenius norm algorithm follows by inspecting the matrix Bernstein inequalities (see Lemma \ref{lem:vector-bernstein-inequality}) and taking $c = L = \big(\sum_i  \|[\mat{A}]_{\bullet, i}\|_2 \|[\mat{B}]_{i, \bullet}\|_2\big)^2 \leq \norm{\mat{A}}_{F}^2\norm{\mat{B}}_{F}^2/\varepsilon^2$. For the max-norm algorithm we can write

\begin{equation*}
    \max_{i} [\mat{\Sigma}]_{ii}  \leq  \sum_{k} \max_{ij}  [\mat{A}]_{ik}^2 [\mat{B}]_{k j}^2 / p_k - \min_{ij}(\mat{AB})^2_{i j}  \,.
\end{equation*}

By the method of Lagrange multipliers, or the Cauchy-Schwarz inequality, one can show that the minimizing probability distribution for this quantity is the one with probabilities $p_k \propto \max_{ij}  [\mat{A}]_{ik} [\mat{B}]_{k j}$. For this choice of probability distribution we then have 

\begin{align}
    \max_{i} [\mat{\Sigma}]_{ii} &\leq  \Big(\sum_{k} \max_{ij}  |[\mat{A}]_{ik} [\mat{B}]_{k j}|\Big)^2 - \min_{ij}(\mat{AB})^2_{i j}  \nonumber \\
    &=  \Big(\sum_{k} \max_{i}  |[\mat{A}]_{ik}|  \max_{j} |[\mat{B}]_{k j}|\Big)^2 - \min_{ij}(\mat{AB})^2_{i j}  \nonumber \\
    &\leq \Big(\sum_{k} \max_{i}  |[\mat{A}]_{ik}|^2 \Big) \Big(\sum_{\ell}  \max_{j} |[\mat{B}]_{\ell j}|^2\Big) - \min_{ij}(\mat{AB})^2_{i j}   \nonumber \\
    &=   \|\mat{A}^T\|_{\infty,2}^2  \|\mat{B}\|_{\infty,2}^2  - \min_{ij}(\mat{AB})^2_{i j} \,.
\end{align}
where the second inequality is due to Cauchy-Schwarz. The max-norm algorithm follows by choice of $c = L = \norm{\mat{A}}_{\infty,2}^2\norm{\mat{B}^T}_{\infty,2}^2/\varepsilon^2$ and inspecting the max-norm matrix Bernstein inequality of Lemma \ref{lem:vector-bernstein-inequality}.

Finally, we note that we can bring the complexity down with respect to dependencies on the matrix dimension, by observing that the computational cost of evaluating $\mat{A}\mat{S}^T \mat{S} \mat{B}$ is $O(nm c^{\omega-2})$ due to the sparsity of $\mat{S}$ (first evaluate $\mat{A}\mat{S}^T$ and $\mat{S} \mat{B}$, taking time $O(nc)$ and $O(mc)$ respectively, and then evaluate their product using fast matrix multiplication), and there is no dependence on the inner dimension.                                            
\end{proof}

\subsection{Quantum algorithm}



In this section we quantize the tug-of-war and column-sample matrix sketches. Here we coherently construct $\mat{A}\vec{s}\vec{s}^T\mat{B}$ for given $\vec{s} \in \mathbb{R}^q$ from the tug-of-war sketch (see Section \ref{sec:tug-of-war}) and use quantum multivariate mean estimation to estimate $\E[\mat{A}\vec{s}\vec{s}^T\mat{B}]$. As the algorithm is effectively computing matrix-vector products, it is unclear how to exploit fast matrix multiplication in this setting, unlike in the classical algorithms discussed in the previous section which can bunch together samples. Both algorithms have the same complexity for square matrices, but the quantum algorithm based on the column-sample sketch uses a QROM as the sketching vectors are constructed from column and row norms of the matrices. When considering rectangular matrices, ability to use QROM yields an advantage for the column-sample sketch over the tug-of-war sketch for matrix products with large inner dimension ($\mat{A}$ ``wide", $\mat{B}$ ``tall"), as we can exploit the sparsity of the sketching vectors. 

\begin{theorem}[Quantum tug-of-war algorithm, 2 matrices]\label{thm:quantum-sarlos-2}
    Let $\mat{A} \in \mathbb{R}^{n \times q}$, $\mat{B} \in \mathbb{R}^{q \times m}$. We have a quantum algorithm which works in the circuit model and returns a matrix $\mat{C}\in \mathbb{R}^{n \times m}$, with success probability at least $(1-\delta)$, which satisfies 
    
    \begin{equation}
        \|\mat{C} - \mat{AB} \|_{\max} \leq \varepsilon \quad \text{in time}\quad \widetilde{O}\!
        \left( (nm + mq + nq) \left(\frac{\norm{\mat{A}}_{2,2}\norm{\mat{B}}_{2,2}}{\varepsilon} \right) \log\left( \frac{1}{\delta}\right)\right) \,,
    \end{equation}
\end{theorem}

\begin{proof}
In order to perform quantum multivariate mean estimation (as outlined in Section \ref{sec:quantum-primitives}), we use the following oracles. The action of the binary oracle may be written as

\[ O_{\mat{X}} : \ket{\vec{s}}\ket{0} \mapsto \ket{\vec{s}}\ket{\mat{A}\vec{s} \vec{s}^T  \mat{B}}\,,\]

which can be obtained in $\widetilde{O}(nm + mq + qn)$ time, given QROM access to $\mat{A}$ and $\mat{B}$. Note that we will use the first register in superposition, but we do not need fast coherent read or write for the entries of $\mat{A}$ and $\mat{B}$ -- thus with a binary encoding of both we can also achieve the same runtime in the circuit model, without QROM. For the probability oracle, we relax the 4-wise independence assumption, and allow our event space to include all signed bitstrings $\Omega = \{-1, +1 \}^n$. Here, the probability oracle has action:

\[ U_{\mathbb{P}}: \ket{0} \xrightarrow[]{\mat{H}^{\otimes n}} \frac{1}{\sqrt{2^n}} \sum_{2\vec{s}'-\mathbf{1}_n \in \Omega} \ket{\vec{s}'}\,.  \]

which is instantiated simply by the Hadamard operation in constant time. Finally, we can write the action of the inner product oracle as 
\[  O_{\mat{Y},\mat{X}}: \ket{\mat{Y}}\ket{\vec{s}}\ket{0} \mapsto \ket{\mat{Y}}\ket{\vec{s}}\ket{\Tr[\mat{Y}^T \mat{A}\vec{s} \vec{s}^T \mat{B}]}\,, \]
for chosen instances $\mat{Y} \in (-\frac{1}{2}, \frac{1}{2})^{n \times m}$ from the mean estimation algorithm of Theorem \ref{thm:multivariate-mean-estimation}. We remark that the Hilbert-Schmidt (trace) inner product we denote above is entirely equivalent to the vector inner product, should $\mat{Y}$ and $\mat{A}\vec{s} \vec{s}^T \mat{B}$ be unfurled into column vectors. This inner product can be instantiated as

\begin{align*}
   \ket{\mat{Y}}\ket{\vec{s}}\ket{0}\ket{0} &\xrightarrow[]{O_{\mat{X}}} \ket{\mat{Y}}\ket{\vec{s}}\ket{\mat{A}\vec{s} \vec{s}^T  \mat{B}}\ket{0} \\&\longrightarrow \ket{\mat{Y}}\ket{\vec{s}}\ket{\mat{A}\vec{s} \vec{s}^T  \mat{B}}\ket{\Tr[\mat{Y}^T \mat{A}\vec{s} \vec{s}^T \mat{B}]} \xrightarrow[]{O_{\mat{X}}^{\dag}} \ket{\mat{Y}}\ket{\vec{s}}\ket{0}\ket{\Tr[\mat{Y}^T \mat{A}\vec{s} \vec{s}^T \mat{B}]}\,, 
\end{align*}
where the second step costs $O(nm)$. The quantum multivariate mean estimation routine of Theorem \ref{thm:quantum-mean-estimation} can thus be instantiated with complexity $\widetilde{O}((nm + nq + mq) N)$, attaining max-norm error $\sqrt{\Tr[\mat{\Sigma}]}\log(nm/\delta)/N$ with success probability at least $(1-\delta)$, where $\Tr[\mat{\Sigma}]$ is given in Lemma \ref{lem:tug-of-war-sketch}. The theorem statement is satisfied by picking $N = \sqrt{2}\|\mat{A}\|_F \|\mat{B}\|_F \log(nm/\delta) / \varepsilon$.
\end{proof}

\begin{theorem}[Quantum column-sample algorithm, 2 matrices]\label{thm:quantum-dkm}
    Let $\mat{A} \in \mathbb{R}^{n \times q}$, $\mat{B} \in \mathbb{R}^{q \times m}$. We have a quantum algorithm in the QROM model which returns a matrix $\mat{C}\in \mathbb{R}^{n \times m}$, with success probability at least $(1-\delta)$, which satisfies 
    
    \begin{equation}
        \|\mat{C} - \mat{AB} \|_{\max} \leq \varepsilon \quad \text{in time}\quad \widetilde{O}\!
        \left( nm\left(\frac{\norm{\mat{A}}_{2,2}\norm{\mat{B}}_{2,2}}{\varepsilon} \right) \log\left( \frac{1}{\delta}\right)\right) \,.
    \end{equation}
\end{theorem}

\begin{proof}
    The analysis follows in the same way as in the proof of Theorem \ref{thm:quantum-sarlos-2}, now with the column-sample sketch and equivalent bound for $\Tr[\mat{\Sigma}]$ given in Lemma \ref{lem:covariance-dkm}. However, now $\vec{s}$ is constructed from a probability distribution based on column (row) norms of $\mat{A}$ ($\mat{B}$). Thus, we require a QROM for those norms to instantiate the binary oracle $O_{\mat{X}} : \ket{\vec{s}}\ket{0} \mapsto \ket{\vec{s}}\ket{\mat{A}\vec{s} \vec{s}^T  \mat{B}}$. Compared to Theorem \ref{thm:quantum-sarlos-2}, we can improve the runtime to only depend on outer dimensions $n,m$ by exploiting the sparsity of $\vec{s}$ in the column-sample sketch: $\mat{As}$ only costs $\widetilde{O}(n)$ time to evaluate and $\vec{s}^T\mat{B}$ only costs $\widetilde{O}(m)$ time to evaluate given QROM access to $\mat{A}$ and $\mat{B}$, and their outer product costs $\widetilde{O}(nm)$ time.  The rest of the analysis follows as in the proof of Theorem \ref{thm:quantum-sarlos-2}, and the quantum multivariate mean estimation routine of Theorem \ref{thm:quantum-mean-estimation} can thus be instantiated with complexity $\widetilde{O}((nm) N)$ with the stated error guarantee attained using the same choice of $N$.
\end{proof}

\subsection{Multiple matrices}

Now let's look at the general case of the product of $k$ matrices which we label as $\mat{A}_{1}\cdots \mat{A}_{k}$, with $\mat{A}_\ell \in \Rbb^{n_{\ell} \times n_{\ell+1}}$ for all $\ell \in [k]$. We approximate this by inserting $\vec{s}_{\ell}\vec{s}^T_{\ell}$ in between each matrix, that is we evaluate $\left(\Pi_{\ell=1}^{k-1} \mat{A}_{\ell}\vec{s}_{\ell}\vec{s}^T_{\ell}\right) \mat{A}_{k}$. This product costs $O(k \max_{i\neq j} n_i n_j )$ to evaluate classically (or quantumly, using coherent arithmetic). As for 2 matrices, we will estimate the mean of this quantity, and can repeat the analysis as before, with both quantum and classical algorithms analyzed through the lens of mean estimation, if we can bound the relevant covariances.

\begin{lemma}[Trace of covariance, sketched multiple matrix product]\label{lem:covariance-multiple-matrices}
Consider the matrix-valued random variable $\left(\Pi_{\ell=1}^{k-1} \mat{A}_{\ell}\vec{s}_{\ell} \vec{s}^{T}_{\ell}\right) \mat{A}_{k}$ with each $\vec{s}_{\ell} \in \mathbb{R}^{n_{{\ell}+1}}$ drawn independently. We have
\begin{align} \label{eq:trace-covariance-multi}
    \Tr[\mat{\Sigma}] &= \sum_{\substack{r_0=u_0,  \\ q_{k}=t_{k}}} \sum_{\substack{(q_1, r_1),..., (q_{k-1}, r_{k-1})  \\ (t_1,u_1),..., (t_{k-1},u_{k-1})}}\, \prod_{\ell=1}^{k}[\mat{A}_{\ell}]_{r_{\ell -1} q_{\ell}}[\mat{A}_{\ell}]_{u_{\ell -1} t_{\ell}} \cdot \nonumber \\
    &\hspace{10em}\cdot \left( \prod_{\ell=1}^{k-1} \left( \Cov\Big[    s^{(\ell)}_{q_{\ell}} s^{(\ell)}_{r_{\ell}}   ,  s^{(\ell)}_{t_{\ell}} s^{(\ell)}_{u_{\ell}}    \Big] + \delta_{q_{\ell} r_{\ell}}\delta_{t_{\ell}u_{\ell} } \right) - \prod_{\ell=1}^{k-1} \delta_{q_{\ell} r_{\ell}}\delta_{t_{\ell} u_{\ell} }  \right)
     \,,
\end{align}
where $s^{(\ell)}_j$ denotes the $j$th entry of $\vec{s}^{(\ell)}$.
\end{lemma}

\begin{proof}
We will use the fact that for a product of independent random variables we have that 

\begin{align*}
    \Var\big[\prod_i^{k} X_i\big] &= \prod_i^{k}\left(\Var[X_i]+ \E[X_i]^2 \right) - \prod_i^{k} \E[X_i]^2\,,\\
    \Cov\big[\prod_i^{k} X_i, \prod_i^{k} Y_i\big] &= \prod_i^{k}\left(\Cov[X_i, Y_i]+ \E[X_i] \E[Y_i] \right) - \prod_i^{k} \E[X_i]\E[Y_i]
\end{align*}

Using this, we can write

\begin{align*} 
    \Tr[\mat{\Sigma}] &=  \sum_{ij}\Var\Big[\Big(\prod_{\ell=1}^{k-1} \mat{A}_{\ell} \vec{s}_{\ell} \vec{s}^T_{\ell} \mat{A}_{k} \Big)_{ij}\Big] \nonumber \\ 
    &=  \sum_{ij}\Var\Big[\sum_{(q_1, r_1),..., (q_{k-1}, r_{k-1})}  \prod_{\ell=1}^{k-1}[\mat{A}_{\ell}]_{r_{\ell -1} q_{\ell}} s^{(\ell)}_{q_{\ell}} s^{(\ell)}_{r_{\ell}}   [\mat{A}_{k}]_{r_{k-1} j} \Big] \nonumber \\
    &=  
    \sum_{\substack{r_0=u_0,  \\ q_{k}=t_{k}}} \sum_{\substack{(q_1, r_1),..., (q_{k-1}, r_{k-1})   \\ (t_1,u_1),..., (t_{k-1},u_{k-1})}}\, \prod_{\ell =1}^{k}[\mat{A}_{\ell}]_{r_{\ell -1} q_{\ell}}[\mat{A}_{\ell}]_{u_{\ell -1} t_{\ell}} \Cov\Big[  \prod_{\ell=1}^{k-1}  s^{(\ell)}_{q_{\ell}} s^{(\ell)}_{r_{\ell}}   , \prod_{\ell'=1}^{k-1}  s^{(\ell')}_{t_{\ell'}} s^{(\ell')}_{u_{\ell'}}    \Big] \nonumber \\
    &=  
    \sum_{\substack{r_0=u_0,  \\ q_{k}=t_{k}}} \sum_{\substack{(q_1, r_1),..., (q_{k-1}, r_{k-1})  \\ (t_1,u_1),..., (t_{k-1},u_{k-1})}}\, \prod_{\ell=1}^{k}[\mat{A}_{\ell}]_{r_{\ell -1} q_{\ell}}[\mat{A}_{\ell}]_{u_{\ell -1} t_{\ell}} \cdot \\
    &\hspace{10em}\cdot \left( \prod_{\ell=1}^{k-1} \left( \Cov\Big[    s^{(\ell)}_{q_{\ell}} s^{(\ell)}_{r_{\ell}}   ,  s^{(\ell)}_{t_{\ell}} s^{(\ell)}_{u_{\ell}}    \Big] + \delta_{q_{\ell} r_{\ell}}\delta_{t_{\ell}u_{\ell} } \right) - \prod_{\ell=1}^{k-1} \delta_{q_{\ell} r_{\ell}}\delta_{t_{\ell} u_{\ell} }  \right)  \nonumber
     \,,
\end{align*}

where in the third line we have assigned the indices $r_0 = u_0 = i, q_k = t_k = j$.
\end{proof}

\begin{lemma}[Tug-of-war variances, multiple matrices]\label{lem:covariance-dkm-multiple}
    Adopt the setting of Lemma \ref{lem:covariance-multiple-matrices}, using tug-of-war sketch matrices. We have 
    \begin{align*}
        \Tr[\mat{\Sigma}] \leq 3^{k}{\norm{\vec{A}_1}_{2,2}^2 \cdots \norm{\vec{A}_k}_{2,2}^2}\,.
    \end{align*}
\end{lemma}

\begin{proof}
    From Lemma \ref{lem:tug-of-war-sketch} recall that for the tug-of-war sketch we have $\Cov\big[s_is_j, s_{k}s_{\ell}\big] = (\delta_{ik}\delta_{j \ell} + \delta_{i\ell} \delta_{jk}) (1- \delta_{ij})$. As these correspond only to variances, they are always positive (if we upper bound all coefficients by their positive counterparts) and we can write $\Cov\big[s_is_j, s_{k}s_{\ell}\big] \leq (\delta_{ik}\delta_{j \ell} + \delta_{i\ell} \delta_{jk})$. This can be substituted into the expression for $\Tr[\mat{\Sigma}]$ in Lemma \ref{lem:covariance-multiple-matrices}, giving
    \begin{align*} 
        \Tr[\mat{\Sigma}] &\leq \sum_{\substack{r_0=u_0,  \\ q_{k}=t_{k}}} \sum_{\substack{(q_1, r_1),..., (q_{k-1}, r_{k-1})  \\ (t_1,u_1),..., (t_{k-1},u_{k-1})}}\, \prod_{\ell=1}^{k}[\ols{\mat{A}_{\ell}}]_{r_{\ell -1} q_{\ell}}[\ols{\mat{A}_{\ell}}]_{u_{\ell -1} t_{\ell}} \cdot  \prod_{\ell'=1}^{k-1} \left( \delta_{q_{\ell'}t_{\ell'}}\delta_{r_{\ell'}u_{\ell'}} + \delta_{q_{\ell'}u_{\ell'}} \delta_{r_{\ell'}t_{\ell'}} + \delta_{q_{\ell} r_{\ell}}\delta_{t_{\ell}u_{\ell} } \right)   \nonumber \\
        &= \sum_{\substack{r_0,q_{k}}} \sum_{\substack{(q_1, r_1),..., (q_{k-1}, r_{k-1})}}\, \left( [\ols{\mat{A}_{1}}]_{r_{0} q_{1}}^2 + [\ols{\mat{A}_{1}}]_{r_{0} q_{1}}[\ols{\mat{A}_{1}}]_{r_{0} r_{1}} + [\ols{\mat{A}_{1}}]_{r_{0} q_{1}}[\ols{\mat{A}_{1}}]_{r_{0} r_{1}} \right) \cdot \\
        &\qquad\qquad \cdot \prod_{\ell=2}^{k-1} \left( [\ols{\mat{A}_{\ell}}]_{r_{\ell -1} q_{\ell}}^2 + [\ols{\mat{A}_{\ell}}]_{r_{\ell -1} q_{\ell}}[\ols{\mat{A}_{\ell}}]_{q_{\ell -1} r_{\ell}} + [\ols{\mat{A}_{\ell}}]_{q_{\ell -1} q_{\ell}}[\ols{\mat{A}_{\ell}}]_{r_{\ell -1} r_{\ell}} \right) \cdot \\
        &\qquad\qquad\qquad  \cdot \left( [\ols{\mat{A}_{k}}]_{r_{k -1} q_{k}}^2 + [\mat{A}_{k}]_{r_{\ell -1} q_{k}}[\ols{\mat{A}_{k}}]_{q_{k -1} q_{k}} + [\ols{\mat{A}_{k}}]_{q_{k -1} q_{k}}[\ols{\mat{A}_{k}}]_{r_{\ell -1} q_{k}} \right)
        \\ 
        &\leq  \sum_{\substack{r_0,q_{k}}} \sum_{\substack{(q_1, r_1),..., (q_{k-1}, r_{k-1})}}\, \left( [\ols{\mat{A}_{1}}]_{r_{0} q_{1}}^2 + [\ols{\mat{A}_{1}}]_{r_{0} q_{1}}[\ols{\mat{A}_{1}}]_{r_{0} r_{1}} + [\ols{\mat{A}_{1}}]_{r_{0} q_{1}}[\ols{\mat{A}_{1}}]_{r_{0} r_{1}} \right) \cdot \\
        &\qquad\qquad \cdot \prod_{\ell=2}^{k-1} \left( [\ols{\mat{A}_{\ell}}]_{r_{\ell -1} q_{\ell}}^2 + [\ols{\mat{A}_{\ell}}]_{r_{\ell -1} q_{\ell}}[\ols{\mat{A}_{\ell}}]_{q_{\ell -1} r_{\ell}} + [\ols{\mat{A}_{\ell}}]_{q_{\ell -1} q_{\ell}}[\mat{A}_{\ell}]_{r_{\ell -1} r_{\ell}} \right) \cdot \\
        &\qquad\qquad\qquad  \cdot \left( [\ols{\mat{A}_{k}}]_{r_{k -1} q_{k}}^2 + [\ols{\mat{A}_{k}}]_{r_{\ell -1} q_{k}}[\ols{\mat{A}_{k}}]_{q_{k -1} q_{k}} + [\ols{\mat{A}_{k}}]_{q_{k -1} q_{k}}[\ols{\mat{A}_{k}}]_{r_{\ell -1} q_{k}} \right)
        \\ 
        &\leq  \sum_{\substack{r_0,q_{k}}} \sum_{\substack{(q_1, r_1),..., (q_{k-1}, r_{k-1}) \\ (q_1', r_1'),..., (q_{k-1}', r_{k-1}')}}\, \left( [\ols{\mat{A}_{1}}]_{r_{0} q_{1}'}^2 + [\ols{\mat{A}_{1}}]_{r_{0} q_{1}'}[\ols{\mat{A}_{1}}]_{r_{0} r_{1}'} + [\ols{\mat{A}_{1}}]_{r_{0} q_{1}'}[\ols{\mat{A}_{1}}]_{r_{0} r_{1}'} \right) \cdot \\
        &\qquad\qquad \cdot \prod_{\ell=2}^{k-1} \left( [\ols{\mat{A}_{\ell}}]_{r_{\ell -1} q_{\ell}'}^2 + [\ols{\mat{A}_{\ell}}]_{r_{\ell -1} q_{\ell}'}[\ols{\mat{A}_{\ell}}]_{q_{\ell -1} r_{\ell};} + [\ols{\mat{A}_{\ell}}]_{q_{\ell -1} q_{\ell}'}[\mat{A}_{\ell}]_{r_{\ell -1} r_{\ell}'} \right) \cdot \\
        &\qquad\qquad\qquad  \cdot \left( [\ols{\mat{A}_{1}}]_{r_{k -1} q_{k}}^2 + [\ols{\mat{A}_{k}}]_{r_{\ell -1} q_{k}}[\ols{\mat{A}_{k}}]_{q_{k -1} q_{k}} + [\ols{\mat{A}_{k}}]_{q_{k -1} q_{k}}[\ols{\mat{A}_{k}}]_{r_{\ell -1} q_{k}} \right)
\end{align*}
where in the second line we have done a relabeling in the third term in each bracket $u_{\ell} \rightarrow r_{\ell}$, and in the final line we have decoupled sums using the fact that all terms are positive. We note that the expression in the final line is a sum of products of squared Frobenius norms of different groupings of matrices from $\mat{A}_1$ to $\mat{A}_k$. Using the submultiplicativity of Frobenius norms, and by counting terms, we have 
\begin{equation*}
    \Tr[\mat{\Sigma}] \leq 3^k \|\mat{A}_1\|_F \cdots \|\mat{A}_k\|_F\,,
\end{equation*}
as required.
\end{proof}

\begin{theorem}[Sketching algorithms for multiple matrices]\label{thm:Sarlos-multi}
    Consider a matrix product of $k$ matrices $\mat{A}_{1}\cdots \mat{A}_{k}$, with $\mat{A}_\ell \in \Rbb^{n_{\ell-1} \times n_\ell}$ for all $\ell \in [k]$. With classical preprocessing time linear in the size of the input up to polylogarithmic factors, we give: 
    
    \noindent A classical algorithm which gives
    \begin{equation}
        \|\mat{C} - \mat{AB} \|_{F} \leq \varepsilon \quad \text{in time}\quad \widetilde{O} \left(3^{k}\cdot \max_{i\neq j} n_i n_j \cdot \frac{\norm{\vec{A}_1}_{2,2}^2 \cdots \norm{\vec{A}_k}_{2,2}^2}{\varepsilon^2} \log\left(\frac{1}{\delta} \right)  \right) \,,
    \end{equation}
    and a quantum algorithm which gives
    \begin{equation}
        \|\mat{C} - \mat{AB} \|_{\max} \leq \varepsilon \quad \text{in time}\quad \widetilde{O}\left(3^{k/2}\cdot \max_{i\neq j} n_i n_j \cdot \frac{\norm{\vec{A}_1}_{2,2} \cdots \norm{\vec{A}_k}_{2,2}}{\varepsilon} \log\left(\frac{1}{\delta}\right) \right) \,,
    \end{equation}
    both with success probability at least $(1-\delta)$.
\end{theorem}

\begin{proof}
    The sample complexity can be found by simply using the bound on $\Tr[\mat{\Sigma}]$ derived in Lemma \ref{lem:covariance-multiple-matrices} for the classical matrix Bernstein inequality quoted in Lemma \ref{lem:vector-bernstein-inequality} and the quantum mean estimation subroutine of Theorem \ref{thm:quantum-mean-estimation}, with a choice of $L =  3^k \|\mat{A}_1\|^2_F \cdots \|\mat{A}_k\|^2_F \log(1/\delta)/\varepsilon^2$ and $N =  3^{k/2} \|\mat{A}_1\| \cdots \|\mat{A}_k\|_F \log(n_0n_k/\delta)/\varepsilon$ respectively. The time complexity per sample costs $O(k \max_{i\neq j} n_i n_j )$ for both the classical and quantum algorithm.
\end{proof}

So far we have not exploited fast matrix multiplication. We can instead insert $\mat{S}^T_{\ell}\mat{S}_{\ell}$ in between each matrix. That is, we approximate $\mat{A}_{1}\cdots \mat{A}_{k}$ by evaluating $\left(\Pi_{\ell=1}^{k-1} \mat{A}_{\ell}\mat{S}^T_{\ell}\mat{S}_{\ell}\right) \mat{A}_{k}$. Note that unlike in the case of a product of 2 matrices, there is no longer a simple interpretation of this as a sample mean over choices of $\vec{s}_{\ell}\mat{s}_{\ell}^T$. Thus, now we need to analyze $\Tr[\mat{\Sigma}]$ for the matrix-valued random variable $\left(\Pi_{\ell=1}^{k-1} \mat{A}_{\ell}\mat{S}^T_{\ell}\mat{S}_{\ell}\right) \mat{A}_{k}$ explicitly.

\begin{theorem}[Sketching algorithm for multiple matrices exploiting fast matrix multiplication]\label{thm:Sarlos-multi-fast}
    Consider a matrix product of $k$ matrices $\mat{A}_{1}\cdots \mat{A}_{k}$, with $\mat{A}_\ell \in \Rbb^{n_{\ell-1} \times n_\ell}$ for all $\ell \in [k]$. With classical preprocessing time linear in the size of the input up to polylogarithmic factors, we give a classical algorithm which returns $\mat{C}$ satisfying
    \begin{equation}
       \|\mat{C} - \mat{AB} \|_{F} \leq \varepsilon \quad \text{in time}\quad \widetilde{O} \left(3^{k(\omega-2)}\cdot \max_{i\neq j} n_i n_j \left(\frac{\norm{\vec{A}_1}_{2,2}^2 \cdots \norm{\vec{A}_k}_{2,2}^2}{\varepsilon^2} \right)^{\omega-2} \log\left(\frac{1}{\delta} \right) \right) \,,
    \end{equation}
    with success probability at least $(1-\delta)$.
\end{theorem}

\begin{proof}
    We need to modify our prior analysis for the new random variable $\left(\Pi_{\ell=1}^{k-1} \mat{A}_{\ell}\mat{S}^T_{\ell}\mat{S}_{\ell}\right) \mat{A}_{k}$. We can check that the trace of covariance satisfies 
    \begin{align*}
        \Tr[\mat{\Sigma}] &=  \sum_{ij}\Var\Big[\Big[\prod_{\ell=1}^{k-1} \mat{A}_{\ell} \mat{S}^T_{\ell}\mat{S}_{\ell} \mat{A}_{k} \Big]_{ij}\Big] \\ 
        &= \sum_{ij} \E \bigg[ \Big( \Big[\prod_{\ell=1}^{k-1} \mat{A}_{\ell} \mat{S}^T_{\ell}\mat{S}_{\ell} \mat{A}_{k} \Big]_{ij} - \Big[ \prod_{\ell=1}^{k-1} \mat{A}_{\ell} \Big]_{ij} \Big)^2 \bigg] \\
        &=  \E \bigg[ \sum_{ij} \Big[\prod_{\ell=1}^{k-1} \mat{A}_{\ell} \mat{S}^T_{\ell}\mat{S}_{\ell} \mat{A}_{k} -  \prod_{\ell=1}^{k-1} \mat{A}_{\ell} \Big]_{ij}^2 \bigg] \\
        &= \E \Big[ \Big\| \prod_{\ell=1}^{k-1} \mat{A}_{\ell} \mat{S}^T_{\ell}\mat{S}_{\ell} \mat{A}_{k} -  \prod_{\ell=1}^{k-1} \mat{A}_{\ell} \Big\|_F^2 \Big]\,,
    \end{align*}
    which shows that it controls the expected (squared) error. Second, a modification of Lemma \ref{lem:covariance-multiple-matrices} yields
    \begin{align*} 
    \Tr[\mat{\Sigma}] 
    &= \sum_{\substack{r_0=u_0,  \\ q_{k}=t_{k}}} \sum_{\substack{(q_1, r_1),..., (q_{k-1}, r_{k-1})  \\ (t_1,u_1),..., (t_{k-1},u_{k-1})}}\, \prod_{\ell=1}^{k}[\mat{A}_{\ell}]_{r_{\ell -1} q_{\ell}}[\mat{A}_{\ell}]_{u_{\ell -1} t_{\ell}} \cdot \\
    &\hspace{10em}\cdot \left( \prod_{\ell=1}^{k-1} \left( \Cov\Big[    [\mat{S}^T_{\ell}\mat{S}_{\ell}]_{q_{\ell}, r_{\ell}}   ,  [\mat{S}^T_{\ell}\mat{S}_{\ell}]_{t_{\ell},u_{\ell}}    \Big] + \delta_{q_{\ell} r_{\ell}}\delta_{t_{\ell}u_{\ell} } \right) - \prod_{\ell=1}^{k-1} \delta_{q_{\ell} r_{\ell}}\delta_{t_{\ell} u_{\ell} }  \right)  \nonumber
     \,,
    \end{align*}
    and so similar to before, we can the approximation error by determining the quantities $\Cov\big[    [\mat{S}^T_{\ell}\mat{S}_{\ell}]_{q_{\ell}, r_{\ell}}   ,  [\mat{S}^T_{\ell}\mat{S}_{\ell}]_{t_{\ell},u_{\ell}} \big]$. Recall that $\E[\mat{S}^T \mat{S}] = \mat{I}$. Similar to in Lemma \ref{lem:tug-of-war-sketch} we have
    \begin{align*}
        \Var\big[[\mat{S}^T \mat{S}]_{ii}\big] &= \E\big[[\mat{S}^T \mat{S}\big]_{ii}^2] - 1 \\
        &= \frac{1}{c^2} \E\Big[\sum_{k\ell} h_k(i)^2 h_{\ell}(i)^2 \Big] - 1\\
        &= \sum_{k\ell} \frac{1}{c^2} - 1 \\
        &= 0\,.
    \end{align*}
    The variance of off-diagonal terms ($i\neq j$) satisfy
    \begin{align*}
        \Var\big[[\mat{S}^T \mat{S}]_{ij}\big] &= \E[(\mat{S}^T \mat{S})_{ij}^2] - 0 \\
        &= \frac{1}{c^2}\E\Big[\sum_{k\ell} h_k(i) h_k(j) h_{\ell}(i) h_{\ell}(j) \Big] \\
        &= \frac{1}{c^2}\sum_{k} h_k(i) h_k(j) h_{k}(i) h_{k}(j)  \\
        &= \frac{1}{c} \,.
    \end{align*}
    Thus, we can write $\Var\big[[\mat{S}^T \mat{S}]_{ij}\big] = \frac{1}{c}(1- \delta_{ij})$. Finally, the covariance (for $(i,j) \neq (k,\ell))$ can be evaluated as
    \begin{align}
        \Cov\big[[\mat{S}^T \mat{S}]_{ij}, [\mat{S}^T \mat{S}]_{k \ell}\big] &= \E[(\mat{S}^T \mat{S})_{ij} (\mat{S}^T \mat{S})_{k \ell}] - \E[[\mat{S}^T \mat{S}]_{ij}]\, \E[[\mat{S}^T \mat{S}]_{k \ell}] \\
        &= \frac{1}{c^2}\E\Big[ \sum_{qr} h_q(i) h_q(j) h_r(k) h_r(\ell)  \Big] - \delta_{ij} \delta_{k \ell} \\
        &=  \delta_{ij} \delta_{k \ell} + \frac{1}{c} \delta_{i\ell} \delta_{j k} - \delta_{ij} \delta_{k \ell}\\
        &= \frac{1}{c} \delta_{i\ell} \delta_{j k} \,,
    \end{align}
    So that overall, we have $\Cov\big[[\mat{S}^T \mat{S}]_{ij}, [\mat{S}^T \mat{S}]_{k \ell}\big] = \frac{1}{c}(\delta_{i\ell} \delta_{j k} + \delta_{ik}\delta_{j \ell})(1-\delta_{ij})$, which is an analogous result for the covariances to Lemma \ref{lem:covariance-dkm-multiple}, now with a factor $\frac{1}{c}$. We can retrace the proof of Lemma \ref{lem:covariance-multiple-matrices} to obtain the bound
    \begin{equation*}
        \Tr[\mat{\Sigma}] \leq \frac{3^k \|\mat{A}_1\|_F \cdots \|\mat{A}_k\|_F}{c}\,.
    \end{equation*}
    
    Now we observe that the classical matrix product costs $O(k \max_{i\neq j} n_i n_j c^{{\omega}-2})$ time to evaluate, where we take $c = 3^k\norm{\vec{A}_1}_{2,2}^2 \cdots \norm{\vec{A}_k}_{2,2}^2/\varepsilon^2$ to attain the desired approximation guarantee of $\varepsilon$. The success probability is controlled by Markov's inequality, and we can attain any arbitrary constant success probability by increasing $c$ by a constant factor. By performing $O(\log(1/\delta))$ experiments, we can attain success probability at least $(1-\delta)$ by using a median-of-means strategy. 
\end{proof}

We can substitute this into the expressions in Lemma \ref{lem:covariance-sketch} to give
\begin{align*} 
    \Tr[\mat{\Sigma}]
    &= \sum_{ij} \left(\sum_{k\ell}[\mat{A}]_{ik}^2 [\mat{B}]_{\ell j}^2 + [\mat{AB}]_{ij}^2 - 2\sum_{k\ell}[\mat{A}]_{ik}^2 [\mat{B}]_{k j}^2 \right)\\
    &=   \|\mat{A}\|^2_F \|\mat{B}\|_F^2 + \|\mat{AB}\|^2_F - 2\sum_{ijk} [\mat{A}]_{ik}^2 [\mat{B}]_{kj}^2  \,.
\end{align*}

The variance of off-diagonal terms ($i\neq j$) satisfy
\begin{align*}
    \Var\!\big[s_is_j\big] &= \E[(s_is_j)^2] \\
    &= \E\left[h(i)^2 h(j)^2 \right] \\
    &= 1 \,.
\end{align*}
Thus, we can write $\Var\left[s_is_j\right] = (1- \delta_{ij})$.
Finally, the covariance (for $\{i,j\} \neq \{k,\ell\})$ can be evaluated as
\begin{align*}
    \Cov\!\big[s_is_j, s_{k}s_{\ell}\big] &= \E[s_is_j s_{k}s_{\ell}] - \E\big[s_is_j\big]\, \E[s_{k}s_{\ell}] \\
    &= \E\Big[  h(i) h(j) h(k) h(\ell)  \Big] - \delta_{ij} \delta_{k \ell} \\
    &= \delta_{ij} \delta_{k \ell} - \delta_{ij} \delta_{k \ell}\\
    &= 0 \,.
\end{align*}
This means that overall, we have 
\begin{equation}
    \Cov\!\big[s_is_j, s_{k}s_{\ell}\big] = (\delta_{ik}\delta_{j \ell} + \delta_{i\ell} \delta_{jk}) (1- \delta_{ij})\,. 
\end{equation}
We can substitute this into the expressions in Lemma \ref{lem:covariance-sketch} to give
\begin{align*} 
    \Tr[\mat{\Sigma}]
    &= \sum_{ij} \left(\sum_{k\ell}[\mat{A}]_{ik}^2 [\mat{B}]_{\ell j}^2 + [\mat{AB}]_{ij}^2 - 2\sum_{k\ell}[\mat{A}]_{ik}^2 [\mat{B}]_{k j}^2 \right)\\
    &=   \|\mat{A}\|^2_F \|\mat{B}\|_F^2 + \|\mat{AB}\|^2_F - 2\sum_{ijk} [\mat{A}]_{ik}^2 [\mat{B}]_{kj}^2  \,,
\end{align*}

\section{Complexity comparison}
\label{sec:complexity-comparison}

In this section we give a more detailed complexity comparison, as well as compare our time complexity results against prior work whilst relaxing assumptions about square matrices or inability to perform fast matrix multiplication. In Table \ref{tab:results-full} we present a comparison of time complexities for rectangular matrices, where the matrix multiplication exponent is set to $\omega$. We see that exploiting $\omega < 3$ can lead to greatly improved complexities for sketching-based algorithms.

\begin{table}[!h]
\renewcommand{\arraystretch}{2}\label{tab:results-full}
    \centering
    \begin{tabular}{rr|cc|c|c}
        &&  \multicolumn{2}{c|}{Runtime\ $(\widetilde{O}(\cdot))$} & \multirow{2}{*}{\makecell{Data \\ model}} & \multirow{2}{*}{\makecell{Multiple \\ matrices?}} \\
        && $\|\cdot\|_{\max}$ & $\|\cdot\|_F$&    \\\hline
        \multirow{4}{*}{\rotatebox{90}{Classical}} & \cite{cohen1999approximating} & $\frac{\norm{\ols{\mat{A}}  \ols{\mat{B}}}_{1,2}^2}{\varepsilon^2}$ & $\sqrt{nm}\frac{\norm{\ols{\mat{A}} \ols{\mat{B}}}_{1,2}^2}{\varepsilon^2}$ & RAM &\checkmark  \\
        & \cite{sarlos2006improved} & $n_* n_*' \Big(\frac{\norm{\mat{A}}_{2,\infty}^2\norm{\mat{B}^T}_{2,\infty}^2}{\varepsilon^2}\Big)^{\omega-2} $ & $n_* n_*' \left(\frac{\norm{\mat{A}}_{2,2}^2\norm{\mat{B}}_{2,2}^2}{\varepsilon^2}\right)^{\omega - 2}$ & circuit & \checkmark (Thm.~\ref{thm:Sarlos-multi-fast}) \\
        & \cite{drineas2006fast}  & $ n_* n_*'\Big(\frac{\norm{\mat{A}^T}_{\infty,2}^2\norm{\mat{B}}_{\infty, 2}^2}{\varepsilon^2}\Big)^{\omega-2} $ & $n_* n_*'\left(\frac{\norm{\mat{A}}_{2,2}^2\norm{\mat{B}}_{2,2}^2}{\varepsilon^2}\right)^{\omega - 2}$ & circuit & \\
        & \cite{drineas2006fast}  & $ nm\Big(\frac{\norm{\mat{A}^T}_{\infty,2}^2\norm{\mat{B}}_{\infty, 2}^2}{\varepsilon^2}\Big)^{\omega-2} $ & $nm\left(\frac{\norm{\mat{A}}_{2,2}^2\norm{\mat{B}}_{2,2}^2}{\varepsilon^2}\right)^{\omega - 2}$ & RAM & \\
        & Thm.~\ref{thm:improved-cohen-lewis} & $\frac{\norm{\ols{\mat{A}}\mat{\ols{B}}}_{\infty,\infty}  \norm{\mat{\ols{A}}\mat{\ols{B}}}_{1,1}}{\varepsilon^2}$ & $\frac{\norm{\ols{\mat{A}} \ols{\mat{B}}}_{1,1}^2}{\varepsilon^2}$ & RAM &\checkmark
        \\\hline
        \multirow{3}{*}{\rotatebox{90}{Quantum}}  & \cite{shao2018quantum} & $\frac{\norm{\mat{A}}_{2,1}\norm{\mat{B}^T}_{2,1}}{\varepsilon}$ & $\sqrt{nm}\frac{\norm{\mat{A}^T}_{2,1}\norm{\mat{B}}_{2,1}}{\varepsilon}$ & QROM & \\
        & \makecell[r]{Thms.~\ref{thm:quantum-sarlos-2} \\\emph{\&}\,\ref{thm:Sarlos-multi}} & $n_* n_*'\frac{\norm{\mat{A}}_{2,2}\norm{\mat{B}}_{2,2}}{\varepsilon}$ & & \makecell{quantum \\ circuit} & \checkmark \\
        & \makecell[r]{Thm.~\ref{thm:quantum-dkm}} & $nm\frac{\norm{\mat{A}}_{2,2}\norm{\mat{B}}_{2,2}}{\varepsilon}$ & & QROM &  \\
        & \makecell[r]{Thm.~\ref{thm:quantum-cohen-lewis}}  & $\frac{\norm{\ols{\mat{A}} \ols{\mat{B}}}_{1,1}}{\varepsilon}$  & $\sqrt{nm}\frac{\norm{\ols{\mat{A}}\, \ols{\mat{B}}}_{1,1}}{\varepsilon}$ & QRAM &\checkmark 
    \end{tabular}
    \caption{\textbf{Detailed time complexity comparison.} Here we consider multiplication of two rectangular matrices $\mat{A} \in \mathbb{R}^{n \times q}$,  $\mat{B} \in \mathbb{R}^{q \times m}$. We denote $n_*n_*' = \max(nm, nq, qm)$ as the largest pair of dimensions. In all constructions that require memory usage (i.e., either RAM, QROM, or QRAM), we require a classical preprocessing routine that initializes the memory (and hence runs in the RAM-model), and runs in time linear in the size of the input.}
\end{table}

In order to compare the expressions, we recall the following sequence of inequalities that follow directly from H\"older's inequality (see Lemma \ref{lem:holder}):
\[\frac{1}{\sqrt{n}} \norm{\mat{X}}_{1,1} \leq \left\{\begin{array}{c}
    \norm{\mat{X}}_{1,2} \\
    \norm{\mat{X}^T}_{1,2} \\
    \norm{\mat{X}^T}_{2,1} \\
    \norm{\mat{X}}_{2,1}
\end{array}\right\} \leq \sqrt{n}\norm{\mat{X}}_{2,2} = \sqrt{n}\norm{\mat{X}^T}_{2,2} \leq \left\{\begin{array}{c}
    n\norm{\mat{X}}_{2,\infty} \\
    n\norm{\mat{X}^T}_{2,\infty} \\
    n\norm{\mat{X}^T}_{\infty,2} \\
    n\norm{\mat{X}}_{\infty,2}
\end{array}\right..\]

We prove one more inequality here which is stated in Figure \ref{fig:norms}.

\begin{lemma}
Consider $\mat{A} \in \mathbb{C}^{n \times k}$ and $\mat{B} \in \mathbb{C}^{k \times m}$, and denote $\ols{\mat{A}}$ as the matrix with element-wise absolute values. We have
\begin{equation}
    \| \ols{\mat{A}} \ols{\mat{B}} \|_{1,1} \leq  \| {\mat{A}}\|_{2,1} \| {\mat{B}}^T \|_{2,1}\,.
\end{equation}
\end{lemma}

\begin{proof}
    The proof follows from Cauchy-Schwarz, as
    \[\| \ols{\mat{A}} \ols{\mat{B}} \|_{1,1} = \sum_{j,k,\ell} |[\mat{A}]_{j\ell}| \cdot |[\mat{B}]_{\ell k}| \leq \sum_{j,k} \|[\mat{A}]_{j,\bullet}\|_2 \cdot \|[\mat{B}]_{\bullet,k}\|_2 = \| {\mat{A}}\|_{2,1} \| {\mat{B}}^T \|_{2,1}\,.\qedhere\]
\end{proof}

At a high level, it seems that there is a trade-off between the strength of the computational model, and the run-time of the algorithms for achieving the same precision. Indeed, the random-walk-based algorithms, building on the ideas of Cohen \textit{\&} Lewis, achieve the best run-times, but at the same time they require the strongest computational model assumptions (RAM classically, and QRAM quantumly). On the other hand, the sketching-based algorithms have a somewhat worse run-time, especially if we do not consider fast matrix-multiplication algorithms (i.e.~$\omega = 3$). However, the assumptions on the computational model are a lot weaker, since we only need QROM-access quantumly for the approach based on Drineas et al., and for the approach based on Sarl\'os we can even implement it in the circuit model, both classically and quantumly.

The situation becomes slightly more complicated when we additionally take fast matrix-multiplication algorithms into account, i.e., when $\omega < 3$. It seems that there is little hope to beat any classical sketching-based algorithm with these quantum sketching techniques, if we take $\omega < 2.5$. However, the algorithms that achieve a matrix multiplication constant below $2.5$ are very complicated and introduce very high constant factors, so we believe there is still utility in designing simple classical and quantum algorithms that behave slightly worse asymptotically than the classical state-of-the-art, but might beat most straightforward classical implementations.

\section*{Acknowledgements}

SA was supported in part by the European QuantERA project QOPT (ERA-NET Cofund 2022-25), the French PEPR integrated projects EPiQ (ANR-22-PETQ0007) and HQI (ANR-22-PNCQ-0002), and the French ANR project QUOPS (ANR-22-CE47-0003-01).
AC is supported by a Simons-CIQC postdoctoral fellowship through NSF QLCI Grant No. 2016245.
SW acknowledges funding provided by the Institute for Quantum Information and Matter, an NSF Physics Frontiers Center (NSF Grant PHY-23171100).

\printbibliography

\appendix

\end{document}